%% file: main.tex
\documentclass[
final
]{dmtcs-episciences}

\pdfoutput=1

\usepackage[utf8]{inputenc}
\usepackage{subfigure}

\usepackage[round]{natbib}

\usepackage{amssymb,amsmath,amsthm,amsfonts,stmaryrd}
\usepackage{bbm}
\usepackage{xspace}
\usepackage{enumerate}
\usepackage{mathtools}
\usepackage{diagbox}
\usepackage{multirow}

\usepackage{stmaryrd} 
\SetSymbolFont{stmry}{bold}{U}{stmry}{m}{n}
\usepackage{anyfontsize} 

\usepackage{tikz}
\usetikzlibrary{arrows,automata,trees,backgrounds,decorations.pathmorphing,positioning,calc,shapes.geometric}
\tikzset{shorten >=1pt, >=stealth, auto, node distance=40, initial text=}

\usepackage[capitalize,nameinlink,noabbrev]{cleveref}
\crefname{desc}{Item}{Items}
\Crefname{desc}{Item}{Items}
\DeclareRobustCommand{\abbrevcrefs}{%
\crefname{theorem}{Thm.}{Thms.}%
\crefname{corollary}{Cor.}{Cors.}%
\crefname{proposition}{Prop.}{Props.}%
\crefname{example}{Ex.}{Exs.}%
\crefname{definition}{Def.}{Defs.}%
\crefname{lemma}{Lem.}{Lems.}%
\crefname{figure}{Fig.}{Figs.}%
}
\DeclareRobustCommand{\cabref}[1]{{\abbrevcrefs\cref{#1}}}

\theoremstyle{plain}
\newtheorem{theorem}{Theorem}
\newtheorem{lemma}[theorem]{Lemma}
\newtheorem{corollary}[theorem]{Corollary}
\newtheorem{proposition}[theorem]{Proposition}

\theoremstyle{definition}
\newtheorem{definition}[theorem]{Definition}

\newtheorem{example}[theorem]{Example}
\newtheorem{remark}[theorem]{Remark}

\theoremstyle{remark}

\newtheorem{claim}[theorem]{Claim}

\makeatletter
\newcommand{\charfusion}[2]{%
  \def\ch@rfusion##1##2{%
    \ooalign{\hfil$##1#1$\hfil\cr\hfil$##2$\hfil\crcr}}%
      \mathop{%
      \vphantom{#1}%
      \mathpalette\ch@rfusion#2}\displaylimits}
\makeatother

\newcommand{\cupcdot}{\charfusion{\cup}{\cdot}}

\makeatletter
\newcommand{\fixed@sra}{\ensuremath{\vrule height 2\fontdimen22\textfont2 width 0pt\shortrightarrow}}
\newcommand{\shortarrow}[1]{%
  \mathrel{\text{\rotatebox[origin=c]{\numexpr#1*45}{\fixed@sra}}}
}
\makeatother

\makeatletter
\newcommand{\bigplus}{%
  \DOTSB\mathop{\mathpalette\mattos@bigplus\relax}\slimits@
}
\newcommand\mattos@bigplus[2]{%
  \vcenter{\hbox{%
    \sbox\z@{$#1\sum$}%
    \resizebox{!}{0.9\dimexpr\ht\z@+\dp\z@}{\raisebox{\depth}{$\m@th#1+$}}%
  }}%
  \vphantom{\sum}%
}
\makeatother

\newcommand{\blank}{\ensuremath{\mbox{\protect\raisebox{.65ex}{\ensuremath{\underbracket[.5pt][1.5pt]{\hspace*{1ex}}}}}}}

\newcommand{\inp}{\mathbbm{i}}
\newcommand{\outp}{\mathbbm{o}}

\newcommand{\leqlag}[1]{\leq\!\!#1}

\newcommand{\llag}[1]{<\!\!#1}
\newcommand{\glag}[1]{>\!\!#1}

\newcommand{\SigmaI}{\ensuremath{\Sigma}}
\newcommand{\SigmaO}{\ensuremath{\Gamma}}
\newcommand{\SigmaIO}{\ensuremath{\Sigma\cup\Gamma}}

\newcommand{\Rec}{\ensuremath{\textnormal{\textsc{Rec}}}}
\newcommand{\Aut}{\ensuremath{\textnormal{\textsc{Aut}}}}
\newcommand{\Rat}{\ensuremath{\textnormal{\textsc{Rat}}}}

\newcommand{\all}{\ensuremath{\textnormal{\textsc{Reg}}_{\mathsf{all}}}}
\newcommand{\fsl}{\ensuremath{\textnormal{\textsc{Reg}}_{\mathsf{FSL}}}}
\newcommand{\fs}{\ensuremath{\textnormal{\textsc{Reg}}_{\mathsf{FS}}}}

\newcommand{\fse}{\ensuremath{\mathit{fse}}}
\newcommand{\prefs}[1]{\mathrm{Prefs}(#1)}

\author{Christof Löding\affiliationmark{1}
  \and Sarah Winter\affiliationmark{2}\thanks{This work is supported by the MIS project F451019F (F.R.S.-FNRS). Sarah Winter is a postdoctoral researcher at F.R.S.-FNRS.}}
\title[Resynchronized Uniformization and Definability Problems for Rational Relations]{Resynchronized Uniformization and Definability Problems for Rational Relations\thanks{This work was supported by the DFG grant LO 1174/3.}}
\affiliation{
  RWTH Aachen University, Germany\\
  Université libre de Bruxelles, Belgium}
\begin{document}
\publicationdata{vol. 25:2 }{2023}{9}{10.46298/dmtcs.7460}{2021-05-05;
2021-05-05; 2022-10-18; 2023-04-14}{2023-05-18}
\maketitle
\begin{abstract}
Regular synchronization languages can be used to define rational relations of finite words, and to characterize subclasses of rational relations, like automatic or recognizable relations.
We provide a systematic study of the decidability of uniformization and definability problems for subclasses of rational relations defined in terms of such synchronization languages. We rephrase known results in this setting and complete the picture by adding several new decidability and undecidability results.
\end{abstract}

\section{Introduction}\label{sec:intro}
\input{sections/intro}

\section{Preliminaries}\label{sec:prelims}
\input{sections/prelims}

\section{Uniformization problems}\label{sec:unif}
\input{sections/unif-problems}

\section{Definability problems}\label{sec:def}
\input{sections/definability-problems}

\section{Conclusion}\label{sec:conclusion}
\input{sections/conclusion}


\bibliographystyle{abbrvnat}
\bibliography{biblio}
\label{sec:biblio}

\end{document}

%% file: sections/intro.tex
In this paper we study uniformization and definability problems for subclasses of rational relations over finite words. The class of (binary) rational relations is defined by transducers, that is, nondeterministic finite state automata in which each transition is annotated with a pair of an input and an output word (see \cite{berstel2009,sakarovich:2009a}). One can also view such a transducer as a two-tape automaton that processes two given words with two reading heads. The two heads can process their word at different speeds (for example, in each transition one can read only one letter on the first tape and two letters on the second tape). While the class of rational relations is rather expressive, it does not have good algorithmic and closure properties (see \cite{berstel2009}). For example, the class is not closed under intersection and complement, and the universality problem (whether a given transducer accepts all pairs of words) is undecidable.

One obtains subclasses of rational relations by restricting the way the two heads can move on their tapes. 
An important such subclass of rational relations are synchronous rational relations, often called automatic relations, see \cite{DBLP:journals/tcs/FrougnyS93,DBLP:conf/lcc/KhoussainovN94,DBLP:conf/lics/BlumensathG00}. These are obtained when restricting the transducers to move synchronously on the two tapes, that is, reading in each transition one letter from each word (that has not yet been fully read). Automatic relations enjoy many closure and algorithmic properties that make them, for example, a suitable tool in decision procedures for logic, see \cite{DBLP:conf/lics/BlumensathG00}.

The class of recognizable relations is obtained when processing two words completely asynchronously, that is, the transducer completely reads the word on the first tape, and then it reads the word on the second tape. It is not hard to see that recognizable relations are finite unions of products of regular languages, see \cite[Chapter~III, Theorem~1.5]{berstel2009}. 

The definability problem for two classes $C_1,C_2$ of word relations is the problem of deciding for a given relation in $C_1$ whether it is in $C_2$. For example, if $C_1$ is the class of rational relations, and $C_2$ the class of automatic relations, then the problem is to decide whether a given rational relation can be defined by a synchronous transducer. This problem was shown undecidable by \cite[Proposition~5.5]{DBLP:journals/tcs/FrougnyS93}, and also the definability problem for rational and recognizable relations was shown undecidable by \cite[Chapter~III, Theorem~8.4]{berstel2009}. A systematic study of definability problem for some important subclasses of rational relations is given by \cite{DBLP:journals/ita/CartonCG06}.

In this paper, we consider a more general version of the definability problem, where the subclasses of rational relations are defined by synchronization languages, a formalism that has been introduced and studied by \cite{conf/stacs/FigueiraL14,DBLP:conf/stacs/DescotteFF19}. In that setting, a synchronization language is a language over the alphabet $\{1,2\}$, specifying how the two heads of a transducer are allowed to move on the two tapes. For example, the word $1121222$ indicates that two symbols from the first tape are read, then one symbol from the second tape, then one symbol from the first tape, and then three symbols from the second tape. Each run of a transducer defines such a synchronization word. Therefore, a regular synchronization language $L$ defines a subclass of the rational relations, namely all relations that can be defined by a transducer whose runs have synchronization words in $L$. For example, the class of automatic relations is defined by the synchronization language $(12)^*(1^* + 2^*)$, and the class of recognizable relations by $1^*2^*$, see \cite{conf/stacs/FigueiraL14}.

We consider here a more general version of synchronization languages. Assuming that the input (first tape) alphabet $\SigmaI$, and the output (second tape) alphabet $\SigmaO$ of the relation are disjoint, we consider synchronization languages that are regular subsets of $(\SigmaIO)^*$. A word over $(\SigmaIO)^*$ does not just specify in which order the heads move along the tape, but also which symbols have to be read in each step. This way, it is possible to capture subclasses of rational relations that cannot be captured with the synchronization languages over $\{1,2\}$, as for example the class of prefix-recognizable relations (see \cref{sssec:prefixrec}). Another motivation for studying this more general formalism of synchronization languages is its tight connection to origin semantics.

Origin semantics is motivated by the intractability of, e.g., the equivalence problem for rational relations (whether two given transducers recognize the same relation), which stems from the fact that two transducers recognizing the same relation may produce their outputs very differently. This is witnessed by their different synchronization languages. 
To overcome this difficulty \cite{DBLP:conf/icalp/Bojanczyk14,DBLP:conf/icalp/BojanczykDGP17} have introduced and studied transducers with origin semantics, that is, additionally, there is an origin function that maps output positions to their originating input positions.
The main result of \cite{DBLP:conf/icalp/Bojanczyk14} is a machine-independent characterization of relations defined by two-way transducers with origin semantics \cite[Theorem~1]{DBLP:conf/icalp/Bojanczyk14}.
Using the origin semantics, many interesting problems become decidable, e.g., deciding whether two given (one-way) transducers are equivalent reduces to deciding whether their synchronization languages are equal.
However, the origin semantic is rather rigid.
To mitigate this, \cite{FJLW16} have introduced a similarity measure between (one-way\footnote{The introduction of this measure has triggered similar work on two-way transducers, see \cite{DBLP:conf/fsttcs/BoseMPP18,DBLP:conf/mfcs/BoseKMPP19}.}) transducers with origin semantics which amounts to a similarity measure between their synchronization languages.
Restricting, for example, the equivalence problem to transducers that have `similar' synchronization languages allows to regain decidability. Since origin semantics of transducers cannot be captured by synchronization languages over $\{1,2\}$, this is another motivation for studying the formalism of synchronization languages over $\Sigma \cup \Gamma$. For example, a synchronization language over $\Sigma \cup \Gamma$ can capture all transducers that behave similarly to a given one in the above mentioned measure.

We study here the decidability of the definability problem, whether a given rational, automatic, or recognizable relation can be defined inside a given synchronization language $T$ for different classes of synchronization languages. 
It is easy to prove that all instances of the definability problem are undecidable if the class of target synchronization languages is regular but otherwise not restricted.
Most importantly, we show that it is decidable whether a given automatic relation can be defined inside a given target synchronization language that comes from the class of synchronization languages that define recognizable relations.
It remains open if it is decidable whether a given automatic relation can be defined inside a given target synchronization language that comes from the class of synchronization languages that define automatic relations, but we show decidability for a few special cases. 
These and further results are summarized in \cref{tab:overviewdefinabilty} on \cpageref{tab:overviewdefinabilty}. An explanation of the results is given in \cref{sec:overview-def}. We note here that another kind of definability problem in the context of synchronization languages has been studied by \cite{descotte2018resynchronizing}: Given two synchronization languages (over $\{1,2\}$), are all relations that are definable with the first synchronization language also definable with the second synchronization language? So the definability problem considered by \cite{descotte2018resynchronizing} is on the level of classes of relations, while we study definability for individual relations.

Furthermore, we also study the uniformization problem for subsequential transducers.
A uniformization of a binary relation $R$ is a subset of $R$ that
defines a function and that has the same domain as $R$, so the
function selects precisely one image for each element in the domain.
In our setting, we are given a relation over finite words, and ask whether this relation has a uniformization that can be computed by a subsequential transducer. A subsequential transducer is a transducer that is deterministic on the input symbols, and produces in each step a finite sequence of output symbols (the prefix ``sub'' refers to the fact that the transducer can make a final output at the end of the word). The uniformization problem by synchronous sequential transducers (that read in each step one input symbol and produce one output symbol) has first been studied for automatic relations over infinite words in the context of synthesizing sequential functions from a logic specification by \cite{buechi}. The setting in which the sequential transducer needs not be synchronous has been considered by \cite{hosch1972finite,holtmann2010degrees}.

While the decidability results for uniformization of automatic relations by synchronous sequential transducers can easily be transferred from infinite to finite words, the problem becomes different when the sequential transducer needs not be synchronous. Decidability for uniformization of automatic relations by subsequential transducers has been shown by \cite[Theorem~18]{CarayolL14}, and for other subclasses of rational relations (finite valued and deterministic rational relations) by \cite{FJLW16}. For rational relations, the problem becomes undecidable, see \cite[Theorem~17]{CarayolL14} (however, each rational relation has a uniformization by a rational function, see \cite{sakarovich:2009a}, which needs not be subsequential).

The problem of uniformization by subsequential transducers does not impose any restrictions on how the subsequential transducer produces its output. If one is, for example, interested in uniformization by subsequential transducers that have a bound $k$ on the delay between input and output (the difference of the length of the processed input and the produced output), then this can be captured by using a regular synchronization language as another parameter to the problem. The subsequential transducer then has to produce its output in such a way that the resulting interleaving of input and output is inside the given synchronization language. As an example (of unbounded delay between input and output), transducers that produce exactly two output symbols for each input symbol are captured by the synchronization language $(\SigmaI\SigmaO\SigmaO)^*$.

Uniformization of automatic relations by subsequential transducers with a synchronization language of finite shiftlag (which roughly corresponds to the bounded delay case mentioned in the previous paragraph) has been shown to be decidable by \cite{DBLP:conf/icalp/Winter18}.
Uniformization of rational relations by subsequential transducers with a synchronization language `similar' to the one of the given transducer has been studied in \cite{FJLW16}.

Our contribution is a systematic overview and some new results of uniformization problems for rational, automatic, and recognizable relations by subsequential transducers for different classes of synchronization languages. 
The most involved result is that it is undecidable whether a given automatic relation has a uniformization by a subsequential transducer inside a given regular target synchronization language.
Furthermore, we show is that it is decidable whether a given rational relation has a uniformization by a recognizable relation.
The proof is a reduction to the boundedness problem for distance automata.
A complete overview is given in \cref{tab:overview} on \cpageref{tab:overview}, and the results are explained in more detail in \cref{sec:overview}.

The paper is structured as follows. In \cref{sec:prelims} we provide definitions and terminology used throughout the paper. In \cref{sec:unif} we present results on uniformization, and in \cref{sec:def} we consider definability problems.

%% file: sections/prelims.tex
In this section, we introduce our terminology and recall some results that are used throughout the paper. Let $\mathbbm N$ denote the set of non-negative integers.

\subparagraph{Words, languages, relations, and uniformizations.}

An \emph{alphabet} is a finite non-empty set of \emph{symbols}, also called \emph{letters}.
We denote alphabets by $\Sigma$ and $\Gamma$.
A \emph{word} over $\Sigma$ is a possibly empty finite sequence $a_1\cdots a_n$ of symbols from $\Sigma$.
The empty word is denoted by $\varepsilon$.
Given a word $w$, we denote by $|w|$ its length, that is, the number of symbols in $w$, by $|w|_a$ the number of occurrences of letter $a$ in $w$, by $w[i]$ its $i$th symbol, and by $w[i:j]$ its infix $w[i]\cdots w[j]$ for $i \leq j$.

We denote by $\Sigma^{-1}$ the set of symbols $a^{-1}$ for all $a \in \Sigma$.
Any word $w$ over $\Sigma \cup \Sigma^{-1}$ can be reduced into a unique irreducible word using the free group equations $aa^{-1} = a^{-1}a = \varepsilon$ for all $a \in \Sigma$, e.g., $c^{-1}abb^{-1}a^{-1}a = c^{-1}a$.

Let $\Sigma^*$, and $\Sigma^+$ denote the set of finite and non-empty finite words over $\Sigma$, respectively.
We write $\Sigma^k$ for $\underbrace{\Sigma\cdots\Sigma}_{k\text{ times}}$, and $\Sigma^{\leq k}$ for $\bigcup_{i=0}^{k} \Sigma^i$.
A \emph{language} $L$ is a subset of $\Sigma^*$, its set of prefixes is denoted by $\prefs{L}$, and $u^{-1}L$ is the set $\{v \mid uv\in L\}$.

A (binary) \emph{relation} $R$ is a
subset of $\Sigma^* \times \Gamma^*$, its domain $\mathrm{dom}(R)$ is the set $\{ u \mid (u,v) \in R\}$, and $(u,v)^{-1}R$ is the relation $\{ (x,y) \mid (ux,vy) \in R \}$.

A \emph{uniformization} of a relation $R \subseteq \Sigma^* \times \Gamma^*$ is a total function $f\colon \mathrm{dom}(R) \to \Gamma^*$ such that $(u,f(u)) \in R$ for all $u \in \mathrm{dom}(R)$.

\subparagraph{Synchronization languages.}

We consider relations over $\SigmaI^* \times \SigmaO^*$, and 
refer to $\SigmaI$ as input alphabet and to $\SigmaO$ as output alphabet.
We assume that $\SigmaI$ and $\SigmaO$ are disjoint, if not, we annotate every symbol in $\SigmaI$ with $1$ and every symbol in $\SigmaO$ with $2$ to make them disjoint. A language over $\SigmaIO$ then encodes a relation over $\SigmaI^* \times \SigmaO^*$, as explained in the following.
For $c \in \{\inp,\outp\}$, referring to input and output, respectively, we define two functions $\pi_{c}\colon (\SigmaIO) \rightarrow (\SigmaIO \cup \{\varepsilon\})$ by $\pi_{\inp}(a) = a$, $\pi_{\outp}(a) = \varepsilon$ if $a \in \Sigma$, and $\pi_{\inp}(a) = \varepsilon$, $\pi_{\outp}(a) = a$ if $a \in \SigmaO$.
These functions are lifted to words over $\SigmaIO$ by applying them to each letter.
A word $w \in (\SigmaIO)^*$ is a \emph{synchronization} of the uniquely determined pair $(u,v) \in \SigmaI^* \times \SigmaO^*$ with $u = \pi_{\inp}(w)$ and  $v = \pi_{\outp}(w)$. We write $\llbracket w \rrbracket$ to denote $(u,v)$.
Then a language $S \subseteq (\SigmaIO)^*$ of synchronizations defines the relation $\llbracket S \rrbracket = \{ \llbracket w \rrbracket \mid w \in S\}$.
Such languages over $\SigmaIO$ are called \emph{synchronization languages}.

Given a synchronization language $S \subseteq (\SigmaIO)^*$, we say a word $w \in (\SigmaIO)^*$ is \emph{$S$-controlled} if $w \in S$.
A language $T \subseteq (\SigmaIO)^*$ is \emph{$S$-controlled} if all its words are, namely, if $T \subseteq S$. 

\subparagraph{Lag, shift, and shiftlag.}

Different classes of automaton definable relations can be defined in terms of the interleaving of input and output symbols in synchronization languages, see \cite{conf/stacs/FigueiraL14}. These characterizations make use of the notions lag, shift, and shiftlag, as defined below.
For the examples illustrating these definitions, consider $\Sigma = \{a\}$ and $\Gamma = \{b\}$.
Given a word $w \in (\SigmaIO)^*$, a position $i \leq |w|$, and $\gamma \in \mathbbm{N}$, we say $i$ is \emph{$\gamma$-lagged} if the absolute value of the difference between the number of input and output symbols in $w[1:i]$ is $\gamma$, e.g., the last position $5$ in the word $abaaa$ is $3$-lagged.
Likewise, we define the notions \emph{$\glag{\gamma}$-lagged} and \emph{$\llag{\gamma}$-lagged}.

Given a word $w \in (\SigmaIO)^*$, a \emph{shift} of $w$ is a position $i \in \{1,\dots,|w|-1\}$ such that $w[i]\in\Sigma \leftrightarrow w[i+1]\notin \Sigma$.
Two shifts $i < j$ are \emph{consecutive} if there is no shift $\ell$ such that $i < \ell < j$.

Given a word $w \in (\SigmaIO)^*$, let $\mathit{shift}(w)$ be the number of shifts in $w$, let $\mathit{lag}(w)$ be the maximum lag of a position in $w$, and let $\mathit{shiftlag}(w)$  be the maximum $n \in \mathbbm{N}$ such that $w$ contains $n$ consecutive shifts which are $\glag{n}$-lagged.

As an example consider the word $w = aabaabbbbbbbaaab$: $\mathit{lag}(w)$ is $4$, because position twelve (at the end of the $b$-block) is $4$-lagged; $\mathit{shift}(w)$ is $5$; $\mathit{shiftlag}(w)$ is $2$, because the two consecutive shifts at positions five and twelve are $\glag2$-lagged.

We lift these notions to languages by taking the supremum in $\mathbbm N \cup \{\infty\}$, e.g., $\mathit{shift}(S) = \mathrm{sup}\{\mathit{shift}(w) \mid w \in S\}$, and likewise for $\mathit{lag}(S)$ and $\mathit{shiftlag}(S)$.
For example, the language $a^*b^*a^*$ has infinite lag and shiftlag of $2$, the language $a^*b^*a^*b^*$ has shiftlag of $3$.

\begin{theorem}[{\cite[Theorem~2]{conf/stacs/FigueiraL14}}]
  Let $S \subseteq (\SigmaIO)^*$ be a regular language.
  It is decidable whether 
  \begin{enumerate}
   \item $\mathit{lag}(S)$ is finite,
   \item $\mathit{shift}(S)$ is finite, and 
   \item $\mathit{shiftlag}(S)$ is finite.
  \end{enumerate}
\end{theorem}

\subparagraph{Automata.}

A \emph{non-deterministic finite state automaton (NFA)} is a tuple $\mathcal A = (Q,\Sigma,q_0,\Delta,F)$, where $Q$ is a finite set of \emph{states}, $\Sigma$ is an alphabet, $q_0 \in Q$ is an \emph{initial state}, $\Delta \subseteq Q \times \Sigma \times Q$ is a \emph{transition relation}, and $F \subseteq Q$ is a set of \emph{final states}.
A \emph{run} $\rho$ of $\mathcal A$ from a state $p_0$ to a state $p_n$ on a non-empty word $w = a_1\cdots a_n$ is a non-empty sequence of transitions $(p_0,a_{1},p_{1})(p_1,a_2,p_2)\cdots(p_{n-1},a_n,p_n) \in \Delta^*$.
We write $\mathcal A\colon p_0 \xrightarrow{w} p_n$ if such a run exists.
A \emph{run} on the empty word is a single state, we write $\mathcal A\colon p \xrightarrow{\varepsilon} p$.
A run is \emph{accepting} if it starts in the initial state and ends in a final state.
For each state $q \in Q$, we denote by $\mathcal A_q$ the NFA that is obtained from $\mathcal A$ by setting the initial state to $q$.

The language \emph{recognized} by $\mathcal A$ is $L(\mathcal A) = \{ w\in\Sigma^* \mid \mathcal A\colon q_0 \xrightarrow{w} q \in F\}$.
The class of languages recognized by NFAs is the class of \emph{regular} languages.

An NFA is \emph{deterministic} (a DFA) if for each state $q \in Q$ and $a \in \Sigma$ there is at most one outgoing transition.
In this case, it is more convenient to express $\Delta$ as a (partial) function $\delta: Q \times \Sigma \rightarrow Q$.
Furthermore, let $\delta^*$ denote the usual extension of $\delta$ from symbols to words.

\subparagraph{Transducers.}

A \emph{non-deterministic finite state transducer (NFT)}, \emph{transducer} for short, is a tuple $\mathcal T =
(Q,\Sigma,\Gamma,q_0,\Delta,F,f)$, where $Q$ is finite set of states,
$\Sigma$ and $\Gamma$ are alphabets, $q_0 \in Q$ is an initial
state, $\Delta \subseteq Q \times \Sigma^* \times \Gamma^* \times Q$
is a finite set of transitions, $F \subseteq Q$ is set of final states, and $f\colon F \to \Gamma^*$ is a \emph{final output function}.
A \emph{non-empty run} $\rho$ is a non-empty sequence of transitions $(p_0,u_1,v_1,p_1)(p_1,u_2,v_2,p_2)\cdots (p_{n-1},u_{n},v_{n},p_n) \in \Delta^*$.
The \emph{input} (resp.\ \emph{output}) of $\rho$ is $u = u_0\cdots u_{n}$ (resp.\ $v = v_0\cdots v_{n}$).
Shorthand, we write $\mathcal T\colon p_0 \xrightarrow{u|v} p_n$.
An \emph{empty run} is a single state, we write $\mathcal T\colon p \xrightarrow{\varepsilon|\varepsilon} p$.
A run is \emph{accepting} if it starts in the initial state and ends in a final state.

The relation \emph{recognized} by $\mathcal{T}$ is
$R(\mathcal T) = \{ (u,v\cdot f(q)) \in \Sigma^* \times \Gamma^* \mid \mathcal T\colon q_0 \xrightarrow{u|v} q \in F\}$.
The class of relations recognized by NFTs is the class of \emph{rational} relations.
The transducer is \emph{subsequential} if each transition reads exactly one input symbol and the transitions are deterministic on the input, in other words, if
$\Delta$ corresponds to a (partial) function $Q \times \Sigma \to \Gamma^* \times Q$.
Then it defines a \emph{subsequential function}.

A transducer induces a synchronization language in a natural way.
A synchronization of a pair of words can be seen as a representation of how a transducer consumes the input and output word of a pair.
We denote the synchronization language of a transducer $\mathcal T$ by $S(\mathcal T)$ which is formally defined as the set
\[
\begin{array}{l}
  \{ u_1v_1u_2v_2\cdots u_nv_nf(p_n) \in (\SigmaIO)^* \mid \\ 
  \qquad\qquad\qquad (q_0,u_1,v_1,p_1)(p_1,u_2,v_2,p_2)\cdots (p_{n-1},u_{n},v_{n},p_n) \in \Delta^* \text{ and } p_n \in F\}.
\end{array}
\]
It is easy to see that an NFA for $S(\mathcal T)$ can be obtained from $\mathcal T$.
Note that this directly implies that $S(\mathcal T)$ is regular.
Given a regular language $S \subseteq (\SigmaIO)^*$, we say that $\mathcal T$ is \emph{$S$-controlled} if $S(\mathcal T)$ is $S$-controlled.
We say that $\mathcal T$ has \emph{finite lag} (\emph{shift} resp.\ \emph{shiftlag}) if $S(\mathcal T)$ has finite lag (shift resp.\ shiftlag).

Conversely, for a regular synchronization language $S$ one can easily construct a transducer $\mathcal T$ such that $S(\mathcal T) = S$.



\subparagraph{Characterization of relation classes via synchronization languages.}

In this work, we consider classes of relations that can be defined by regular synchronization languages, that is, subclasses of rational relations. Two standard subclasses of rational relations are automatic relations which are relations that can be defined in terms of automata that synchronously read input and output symbols, and recognizable relations which are relations that can be expressed as finite unions of products of regular languages (see, e.g., \cite{DBLP:journals/ita/CartonCG06}). 
We do not formalize these definitions here because we are working with a different characterization of these relation classes, as explained below.
We denote by \Rec, \Aut, and \Rat\ the classes of recognizable, automatic, and rational relations, respectively.
It is well-known that $\Rec \subsetneq \Aut \subsetneq \Rat$.

We recall the connection between these classes of relations and synchronization languages as established by \cite{conf/stacs/FigueiraL14}.
For a regular synchronization language $S \subseteq (\SigmaIO)^*$, let
\[
  \textnormal{\textsc{Rel}}(S) = \{ \llbracket T \rrbracket \!\mid\! T \subseteq (\SigmaIO)^* \text{ is a regular $S$-controlled language}\}
\]
be the set of relations that can be given by regular $S$-controlled synchronization languages.


\begin{proposition}[\cite{conf/stacs/FigueiraL14}]
The following properties hold.
 \begin{enumerate}
  \item $\textnormal{\textsc{Rel}}(\Sigma^*\Gamma^*) = \Rec$.
  \item $\textnormal{\textsc{Rel}}((\Sigma\Gamma)^*(\Sigma^* + \Gamma^*)) = \Aut$.
  \item $\textnormal{\textsc{Rel}}((\Sigma + \Gamma)^*) = \Rat$.
 \end{enumerate}
\end{proposition}

Let $\all$ be the class of all regular synchronization languages, $\fsl$ and $\fs$ be the class of regular synchronization languages with finite shiftlag and finite shift, respectively.
Clearly, $\fs \subsetneq \fsl \subsetneq \all$.
We have the following result.

\begin{theorem}[\cite{conf/stacs/FigueiraL14}]
 Let $S \subseteq (\SigmaIO)^*$ be a regular language, then:
 \begin{enumerate}
  \item $\textnormal{\textsc{Rel}}(S) \subseteq \Rec$ iff $S \in \fs$,
  \item $\textnormal{\textsc{Rel}}(S) \subseteq \Aut$ iff $S \in \fsl$,
  \item $\textnormal{\textsc{Rel}}(S) \subseteq \Rat$ iff $S \in \all$.
 \end{enumerate}
\end{theorem}

Note that the previous theorem is a statement about the class of relations $\textnormal{\textsc{Rel}}(S)$ defined by all regular subsets of $S$. In particular, if $S \in \fs$ then $\llbracket S \rrbracket \in \Rec$, if $S \in \fsl$ then $\llbracket S \rrbracket \in \Aut$. However, the converse generally does not hold, e.g., consider $S = (\SigmaIO)^*$, clearly $\llbracket S \rrbracket = (\SigmaI^* \times \SigmaO^*) \in \Rec$ but $S \notin \fs$.

Given a class $\mathcal C$ of regular languages over $\SigmaIO$, we say that a regular language $S_0\subseteq (\SigmaIO)^*$ is a \emph{canonical representative} of $\mathcal C$ if $\textnormal{\textsc{Rel}}(S_0) = \bigcup_{S \in \mathcal C} \textnormal{\textsc{Rel}}(S)$.
A canonical representative is called \emph{effective} if for each regular $S$-controlled $T \subseteq (\SigmaIO)^*$ for some $S \in \mathcal C$ there is a regular $S_0$-controlled $T' \subseteq (\SigmaIO)^*$ such that $\llbracket T \rrbracket = \llbracket T' \rrbracket$ and $T'$ can be constructed in finite time.

\begin{theorem}[\cite{conf/stacs/FigueiraL14}]\label{thm:canonical} \label{thm:canonical-representatives}
The following properties hold.
 \begin{enumerate}
  \item $\Sigma^*\Gamma^*$ is an effective canonical representative of $\fs$.
  \item $(\Sigma\Gamma)^*(\Sigma^* + \Gamma^*)$ is an effective canonical representative of $\fsl$.
  \item $(\Sigma + \Gamma)^*$ is an effective canonical representative of $\all$.
 \end{enumerate}
\end{theorem}

%% file: sections/unif-problems.tex
In this section, we take a systematic look at uniformization problems defined by synchronization languages in the classes $\all, \fsl$, and $\fs$. Roughly speaking, these problems are about deciding whether a given relation (given by a synchronization language $S$), has a uniformization by a subsequential transducer whose synchronization language is contained in a given synchronization language $T$.

We begin with an example of a rational relation that is uniformizable by a subsequential transducer.
Afterwards, we formally define our problem settings.

\begin{example}\label{ex:intro2}
Consider the rational relation over the input alphabet $\SigmaI = \{a,b\}$ and the output alphabet $\SigmaO = \{c\}$
\[
  R_1 = \{(aw,c^{|w|_a}) \mid w \in \SigmaI^*\} \cup \{(wb,c^{|w|_b}) \mid w \in \SigmaI^*\}
\]
 defined by the synchronization language $ac(ac + b)^* + (a+bc)^*bc$.
The relation is uniformized by the subsequential transducer depicted in \cref{fig:intro2} on the left-hand side. If the input starts with $a$, then the transducer produces output $c$ for each input $a$ and no output for an input $b$, and vice versa otherwise. 
\end{example}

\input{figures/fig-intro2}

The general problem whether a rational relation has a uniformization by a subsequential transducer is undecidable, see \cite[Theorem~17]{CarayolL14}. We consider the decision problem that introduces a parameter $T$ for the synchronizations that the subsequential transducer can use:

\begin{definition}[Resynchronized uniformization problem]\label{def:resync-unif}
 The \emph{resynchronized uniformization problem} asks, given a regular source language $S \subseteq (\SigmaIO)^*$ and a regular target language $T \subseteq (\SigmaIO)^*$, whether there exists a $T$-controlled subsequential transducer that uniformizes $\llbracket S \rrbracket$.
\end{definition}

We start by giving an example.

\begin{example}\label{ex:intro}
Consider the rational relation over the input alphabet $\SigmaI = \{a,b,c\}$ and output alphabet $\SigmaO = \{d,e\}$
\[
  R_2 = \{(a^iba^j,d(d+e)^k) \mid i,j,k \geq 0 \} \cup \{(a^ica^j,e(d+e)^k) \mid i,j,k \geq 0 \}
\]
defined by the synchronization language $S = da^*ba^*(d+e)^* + ea^*ca^*(d+e)^*$.
Furthermore, consider $T = \SigmaI^*(\SigmaI\SigmaO)^+$.
The relation is uniformized by the $T$-controlled subsequential transducer depicted in \cref{fig:intro2} on the right-hand side.
The uniformization realized by the transducer maps inputs of the form $a^iba^j$ to $d^{j+1}$ and inputs of form $a^ica^j$ to $ed^{j}$ for all $i,j \geq 0$.

$R_2$ can also be uniformized by a subsequential transducer that is $(\SigmaI^*\SigmaO^*)$-controlled.
To be more specific, $R_2$ can be uniformized by a subsequential transducer that is $(\SigmaI^*\SigmaO)$-controlled. Such a transducer can first read the whole input, and then output $d$ if the input is of the form $a^*ba^*$, and output $e$ if the input is of the form $a^*ca^*$.

But $R_2$ cannot be uniformized by a subsequential transducer that is $(\SigmaI\SigmaO)^*$-controlled. Such a transducer would have to make the first output after reading the first input symbol. If this input symbol is $a$, then it is not yet decided if the output has to start with $d$ or with~$e$.
\end{example}

\subsection{Overview of the results}\label{sec:overview}

\begin{table}[t]
\begin{center}
\begin{tabular}{|l|c|c|c|} \hline
 \backslashbox[50mm]{target}{source} & \parbox{25mm}{rational\\$S \in \all$} & \parbox{25mm}{automatic\\$S \in \fsl$} & \parbox{25mm}{recognizable\\$S \in \fs$} \\ \hline 
 \hline

 \parbox{40mm}{\vspace{1mm}unrestricted\\$T = (\SigmaIO)^* \in \all$\vspace{1mm}} &\parbox{22mm}{\vspace{1mm}\centering undec.\\{\centering \cite[Theorem~17]{CarayolL14}}\vspace{1mm}} &  \parbox{22mm}{\vspace{1mm}\centering dec.\\{\centering \cite[Theorem~18]{CarayolL14}}\vspace{1mm}} & {\color{gray}always}\\ \hline

 \parbox{45mm}{\vspace{1mm}synchronous\\$T\! =\! (\SigmaI\SigmaO)^*(\SigmaI^*\! +\! \SigmaO^*)\! \in\! \fsl$\vspace{1mm}} 	& \parbox{20mm}{\centering undec.\\{\centering (\cabref{thm:undec-automatic-sync})}} &  \parbox{26mm}{\vspace{1mm}\centering dec.\\{\centering \cite{buechi}}\vspace{1mm}} & dec. \\ \hline

 \parbox{40mm}{\vspace{1mm}input before output\\$T = \SigmaI^*\SigmaO^*\in \fs$\vspace{1mm}} & \parbox{20mm}{\centering dec.\\{\centering (\cabref{thm:unif-by-rec})}}  &  \parbox{22mm}{\vspace{1mm}\centering dec.\\{\centering \cite[Proposition~20]{CarayolL14}}\vspace{1mm}} & {\color{gray}always}\\ \hline 
 \hline

 $T \in \all$ & undec. & \parbox{20mm}{\vspace{1mm}\centering undec.\\{\centering (\cabref{cor:undec-aut-sync})}\vspace{1mm}} & \parbox{25mm}{\vspace{1mm}\centering dec.\\{\centering\cite[Theorem~25]{DBLP:conf/icalp/Winter18}}\vspace{1mm}}\\ \hline

 $T \in \fsl$ & undec. & \parbox{25mm}{\vspace{1mm}\centering dec.\\{\centering\cite[Theorem~10]{DBLP:conf/icalp/Winter18}}\vspace{1mm}} & dec.\\ \hline
 
 \parbox{40mm}{\vspace{2mm}$T \in \fs$\vspace{1mm}} & open & dec. & dec. \\ \hline 
 \hline

 $\exists T \in \all$ & \multicolumn{3}{c|}{\color{gray}reduces to $T = (\SigmaIO)^*$}\\ \hline

 $\exists T \in \fsl$ & undec. (\cabref{thm:undec-automatic-fsl}) & dec. (\cabref{thm:automatic-fsl}) & {\color{gray}always}\\ \hline

 $\exists T \in \fs$ & dec. (\cabref{cor:unif-finiteshift}) & dec. & {\color{gray}always}\\ \hline
\end{tabular}
\end{center}
\caption[]{Overview of results for instances of the resynchronized uniformization problem (upper and middle rows) and variants (lower rows).
The results are described in \cref{sec:overview}.
}
\label{tab:overview}
\end{table}

\cref{tab:overview} provides an overview of known and new (un)decidability results for the resynchronized uniformization problem for different types of source relations and target synchronization relations.
The entries for the new results contain references to the corresponding statement in this paper, which are proved in \cref{sec:dec-unif} and \cref{sec:undec-unif}. Some entries contain a reference to the literature. The results by \cite[]{DBLP:conf/icalp/Winter18} (see \cite[Theorems 10 and 25]{DBLP:conf/icalp/Winter18}) use the same terminology as in this paper, so they do not need further explanation. Some other results from the literature are not stated in terms of synchronization languages, so we explain their connection to our setting below, and we also briefly explain the table entries without a further reference, which are either trivial or are direct consequences of other entries in the table.

As classes for source relations we consider rational, automatic, and recognizable relations, represented by the columns in the table. Formally, these relations are specified by regular languages from the classes $\all$, $\fsl$, and $\fs$. But note that precise representations of the source relations are not important because according to \cref{thm:canonical-representatives} we can always assume that the source relation is encoded using the canonical synchronizations for rational, automatic, and recognizable relations.

The rows of the table correspond to different variations concerning the target synchronization language $T$ in the problem. The upper three rows show the cases where $T$ is fixed to the canonical synchronizations for rational, automatic, and recognizable relations, respectively.

The first row is about the problem whether a given relation has a $(\SigmaIO)^*$-controlled subsequential uniformization, which is equivalent to asking whether it has a uniformization by an arbitrary subsequential transducer. This problem has been considered by \cite{CarayolL14}:
For rational source relations the problem is undecidable \cite[Theorem~17]{CarayolL14} and for automatic source relations it is decidable \cite[Theorem~18]{CarayolL14}.
The same authors have also shown that it is decidable whether an automatic relation has a uniformization by a recognizable relation, which corresponds to a $\SigmaI^*\SigmaO^*$-controlled subsequential uniformization in our table (third row, middle column), see \cite[Proposition~20]{CarayolL14}.

The question whether an automatic relation has a uniformization by a synchronous subsequential transducer (second row, middle column) is a classical problem of reactive synthesis that has been studied intensively over infinite words and was first solved by \cite{buechi}. The result can easily be adapted from infinite to finite words.

Let us now look at first three rows and the rightmost column (recognizable relations). First, we explain why a recognizable relation always has a uniformization by a $\Sigma^*\Gamma^*$-controlled subsequential transducer, and thus also a by $(\SigmaIO)^*$-controlled one.
  Recall that a recognizable relation $R$ can be effectively represented as $\bigcup_{i = 0}^n U_i \times V_i$ where $U_i \subseteq \SigmaI^*$ and $V_i \subseteq \SigmaO^*$ are regular languages for each $i$. Since regular languages are closed under all Boolean operations, one can find such a representation in which all the $U_i$ are pairwise disjoint.
  Now it suffices to pick some $v_i \in V_i$ for each $i$, then the relation $R_f = \bigcup_{i = 0}^n U_i \times \{v_i\}$ is functional and uniformizes $R$.  
  Clearly, $R_f$ is realizable by a $\Sigma^*\Gamma^*$-controlled subsequential transducer based on a product DFA for the $U_i$.
  Secondly, the result that it is decidable whether a recognizable relation has a uniformization by a synchronous subsequential transducer follows directly from the fact that it is decidable whether an automatic relation has a uniformization by a synchronous sequential transducer because every recognizable relation is also automatic (in general, decidability results propagate to the right in each row of the table).

As already mentioned, the first three rows deal with the fixed target synchronization languages that are the canonical representatives for the classes $\all$, $\fsl$, $\fs$. The three rows in the middle are about the problem where a regular target synchronization language from one of these classes is given as input to the problem. 
The two undecidability results for rational source relations in this part of the table directly follow  for the choices of $T = (\SigmaIO)^*$ resp.\ $T = (\SigmaI\SigmaO)^*(\SigmaI^* + \SigmaO^*)$ from the results in the upper rows. The decidability results in this part of the table follow directly from the decidability results by \cite{DBLP:conf/icalp/Winter18} since $\fs \subseteq \fsl \subseteq \all$.
We conjecture that the case of rational source relations and target synchronization languages of finite shift is decidable, but we were not able to prove it, so this case is still open.

The last three rows consider the version of the problem where the target synchronization language is not fixed or given, but we ask for the existence of a uniformization by a subsequential transducer whose synchronization language is in $\all$, $\fsl$, resp.\ $\fs$. Clearly, for the class $\all$ this reduces to the first row because both problems do not impose any restriction on the synchronization language of the subsequential transducer. Concerning recognizable source relations, we have seen above that a recognizable relation always has $\Sigma^*\Gamma^*$-controlled sequential uniformization, thus, every entry is ``always''. The decidability result for automatic source relations in the last row is a direct consequence of the decidability result one column to the left. 

We now show the remainder of the results, first the decidability results in \cref{sec:dec-unif}, then the undecidability results in \cref{sec:undec-unif}.

\input{sections/unif-decidability}

\input{sections/unif-undecidability}

%% file: figures/fig-intro2.tex
\begin{figure}[t]
    \centering
      \begin{tikzpicture}[thick,node distance=4em]
        \tikzstyle{every state}+=[inner sep=4pt, minimum size=3pt];
        \node[state, initial] (0) {};
        \node[state, right of= 0, accepting] (1) {};
        \node[state, below of= 0, accepting] (2) {};
        \node[state, right of= 2] (3) {};

        \draw[->] (0) edge                node[]        {$a|c$} (1);
        \draw[->] (0) edge                node[swap]  {$b|c$} (2);
        \draw[->] (1) edge[loop above]    node        {$a|c, b|\varepsilon$} ();
        \draw[->] (2) edge[loop left]    node        {$b|c$} ();
        \draw[->] (2) edge[bend left]               node[]  {$a|\varepsilon$} (3);
        \draw[->] (3) edge[bend left]                node[]  {$b|c$} (2);
        \draw[->] (3) edge[loop above]    node        {$a|\varepsilon$} ();
        \draw [->] (1) -- +(2em,0) node[midway]{$\varepsilon$};
        \draw [->] (2) -- +(0,-2em) node[midway,swap]{$\varepsilon$};

        \begin{scope}[xshift=15em]
        \node[state, initial] (0) {};
        \node[state, right of= 0, xshift=2em, accepting] (1) {};

        \draw[->] (0) edge[loop above]    node        {$a|\varepsilon$} ();
        \draw[->] (0) edge                node[]        {$b|d, c|e$} (1);
        \draw[->] (1) edge[loop above]    node[]  {$a|d$} ();
        \draw [->] (1) -- +(2em,0) node[midway]{$\varepsilon$};
        \end{scope}
       \end{tikzpicture}
      \caption{
       The subsequential transducer on the left uniformizes the relation $R_1$ from \cref{ex:intro2},
       the subsequential transducer on the right uniformizes the relation $R_2$ from \cref{ex:intro}.
       }
     \label{fig:intro2}
    \end{figure}
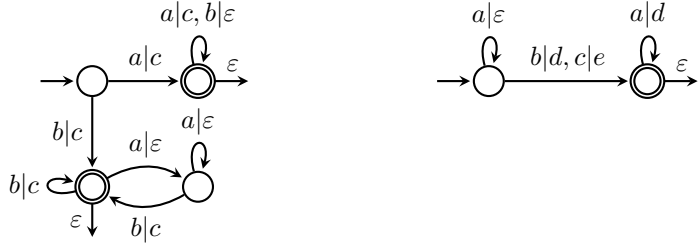 

%% file: sections/unif-decidability.tex
\subsection{Decidability results}\label{sec:dec-unif}

To begin with, we show the following:

\begin{theorem}\label{thm:automatic-fsl}
It is decidable whether a given automatic relation has a uniformization by a subsequential transducer with finite shiftlag.
\end{theorem}

\begin{proof}
This is consequence of the proof that it is decidable whether an automatic relation has a uniformization by a subsequential transducer presented by \cite{CarayolL14}. It turns out that if an automatic relation has a uniformization by a subsequential transducer, then the relation also has a uniformization by a subsequential transducer with finite shiftlag. A bit more formally, \cite{CarayolL14} have reduced the problem to deciding the existence of a winning strategy in a safety game.
Such a winning strategy can be translated into a $(\Sigma^k\Gamma^k)^*\Sigma^*\Gamma^*$-controlled subsequential transducer for a computable~$k$. In other words, as long as the needed lookahead on the input to produce the next output is at most $k$, the subsequential transducer for the uniformization alternates between reading input words of length $k$ and producing output words of length~$k$.
Once the transducer needs a lookahead of more than $k$, it can read the remainder of the input and produce one final output word.
The language $(\Sigma^k\Gamma^k)^*\Sigma^*\Gamma^*$ has finite shiftlag.
\end{proof}

Now we show that it is decidable whether a given rational relation has a uniformization by a $\SigmaI^*\SigmaO^*$-controlled subsequential transducer. The functions computed by $\SigmaI^*\SigmaO^*$-controlled subsequential transducers are clearly recognizable relations, and conversely, each functional recognizable relation can be computed by a $\SigmaI^*\SigmaO^*$-controlled subsequential transducer (as already explained in \cref{sec:overview}).

We prove that it is decidable whether a given rational relation has a uniformization by a recognizable relation by a reduction from the boundedness problem for distance automata.
A distance automaton is a finite state automaton that additionally maps transitions to distances of $0$, $1$, or $\infty$. Then a run is associated with the sum of distances seen along the run.
Such an automaton is said to be bounded if there exists a bound such that each accepted word has a run with a distance less or equal to this bound.
We define this formally.

A \emph{finite state automaton with a distance function} (\emph{distance automaton}) $\mathcal A $ is a tuple $(Q,\Sigma,q_0,\Delta,F,d)$, where the first five components are the same as for finite state automata, and $d\colon \Delta \rightarrow \{0,1,\infty\}$ is a distance function that assigns a value to each transition.
The concepts of run, accepting run, and recognized language are defined as for finite state automata.
The \emph{distance} $d(\rho)$ of a run of the form $(p_0,a_1,p_1)\dots(p_{n-1},a_n,p_n)$ of $\mathcal A$ on $w = a_1\dots a_n$ is defined as $\sum_{i=1}^n d((p_{i-1},a_i,p_i))$, where for any integer $i$, $i < \infty$, and $i + \infty = \infty + i = \infty$.
The \emph{distance} $d(w)$ of a word $w$ is defined as $\infty$ if $w \notin L(\mathcal A)$ and otherwise as the minimal distance of an accepting run of $\mathcal A$ on $w$, that is, $\mathrm{min}\{ d(\rho) \mid \rho \text{ is an accepting run of $\mathcal A$ on $w$}\}$.
Let $D(\mathcal A)$ denote the supremum of distances associated with $L(\mathcal A)$, that is, the supremum of $\{ d(w) \mid w \in L(\mathcal A)\}$.
A distance automaton $\mathcal A$ is \emph{bounded} if $D(\mathcal A)$ is finite.
Hashiguchi has shown that the boundedness problem for distance automata is decidable, see \cite{DBLP:journals/jcss/Hashiguchi82,hashiguchi1990improved}, see also \cite{Kirsten05}.

Now that we have introduced distance automata, we are ready to prove the result.

\begin{theorem}\label{thm:unif-by-rec}
 It is decidable whether a given rational relation has a uniformization by a recognizable relation.
\end{theorem}

\begin{proof}
 We give a reduction to the boundedness problem for distance automata.
 
 Let $R \subseteq \SigmaI^* \times \SigmaO^*$ be a rational relation given by an NFA $\mathcal A$ over $\SigmaIO$ that defines $R$, that is, $\llbracket L(\mathcal{A}) \rrbracket = R$. Our goal is to construct a distance automaton $\mathcal B$ such that such that $\mathrm{dom}(\mathcal B) = \mathrm{dom}(R)$ and $D(\mathcal B)$ is bounded iff $R$ has a uniformization by a recognizable relation.

 Note that a uniformization of $R$ by a recognizable relation only produces finitely many different possible output words. Hence, there is a bound on the length on these output words. And vice versa, if there is a bound $n$ such that for each word $u$ in the domain of $R$ there is an output word $v$ such that $|v| \le n$ and $(u,v) \in R$, then the relation assigning to each input the smallest output in length-lexicographic order is a recognizable uniformization of $R$.

 So we use the reduction to boundedness of distance automata to check such a bound on the length of shortest output words for each input exists. 
 The idea behind the construction is to replace transitions labeled with output symbols with transitions that increase the distance value of a run. Formally, we take all input transitions of $\mathcal{A}$ with distance value $0$, and insert transitions with distance value $1$ that correspond to runs consuming one input symbol and a non-empty sequence of output symbols, as described below. 
 
Let $\mathcal B = (Q_\mathcal A,\SigmaI,q_0^\mathcal A,\Delta,F_\mathcal A,d)$ be a distance automaton, where $\Delta$ and $d$ are defined such that
\[
  (p,a,q) \in \Delta \text{ and } d((p,a,q)) = 0 \text{ iff } \mathcal A\colon p \xrightarrow{a} q \text{ and}
\]
\[
  (p,a,q) \in \Delta \text{ and } d((p,a,q)) = 1 \text{ iff } \mathcal A\colon p \xrightarrow{a} q' \xrightarrow{w} q \text{ for } q' \in Q_{\mathcal{A}}, w \in \SigmaO^+, \text{ and not } \mathcal A\colon p \xrightarrow{a} q.
\]
It is easy to see that $L(\mathcal B) = \mathrm{dom}(R)$. 
We show that $R$ has a uniformization by a recognizable relation iff $D(\mathcal B)$ is bounded.
 
Assume $R$ is uniformized by a recognizable relation $\bigcup_{i=1}^n U_i \times \{v_i\}$.
Let $\ell$ be the length of the longest $v_i$.
Then, for each $u \in \mathrm{dom}(R)$ there exists a $v \in \SigmaO^{\leq \ell}$ such that there is synchronization $w$ of $(u,v)$ that is accepted by $\mathcal A$.
A run of $\mathcal A$ on $w$ can be easily translated into a run of $\mathcal B$ on $u$.
It directly follows that $d(u)$ is at most $\ell$.
Thus, $d(\mathcal B)$ is bounded.

Assume $d(\mathcal B)$ is bounded by $n$.
Thus, for each $u \in \mathrm{dom}(R)$ we can pick a $v \in \SigmaO^*$ of length at most $|\Delta|\cdot n$ such that there exists a synchronization $w$ of $(u,v)$ that is accepted by $\mathcal A$.
The factor $|\Delta|$ was introduced because a number of consecutive transitions with output labels that occur in $\mathcal A$ increase the distance of a corresponding run in $\mathcal B$ only by one.
Since there are only finitely many words of a fixed length it follows directly that $R$ can be uniformized by recognizable relation.

From the decidability of the boundedness problem for distance automata we obtain decidability of the problem whether a rational relation is uniformizable by a recognizable relation.
\end{proof}

Turning to the question whether a rational relation has a subsequential uniformization inside the class of languages with finite shift, we obtain the following.

\begin{corollary}\label{cor:unif-finiteshift}
 It is decidable whether a given rational relation has a uniformization by a subsequential transducer with finite shift.
\end{corollary}

\begin{proof}
 We show that if a rational relation is uniformizable by a subsequential transducer whose synchronization language has finite shift, then it is uniformizable by a recognizable relation which is decidable by \cref{thm:unif-by-rec}.

 Given a rational relation $R$ and a subsequential transducer $\mathcal T$ that uniformizes $R$ such that $S(\mathcal T) \in \fs$.
 Recall that $\Sigma^*\Gamma^*$ is an effective canonical representation of $\fs$.
 Clearly, the $\Sigma^*\Gamma^*$-controlled representation of $S(\mathcal T)$ describes a functional recognizable relation, and as explained earlier, functional recognizable relation can be computed by $\Sigma^*\Gamma^*$-controlled subsequential transducers.
\end{proof}

%% file: sections/unif-undecidability.tex
\subsection{Undecidability results}\label{sec:undec-unif}

First, we introduce a technical result by \cite{CarayolL14} of which \cref{thm:undec-automatic-sync,thm:undec-automatic-fsl} are easy consequences.
There, it is shown that is undecidable whether a given rational relation has a uniformization by a subsequential transducer by showing the following stronger statement.

\begin{lemma}[\cite{CarayolL14}]\label{lemma:automatic}
There exists a rational relation $R_M$, constructed from a Turing machine $M$, such that $R_M$ has a uniformization by the identity function if $M$ does not halt for the empty input, and $R_M$ has no uniformization by any subsequential transducer if $M$ halts for the empty input.
\end{lemma}

Clearly, the identity function can be realized by a $(\Sigma\Gamma)^*$-controlled subsequential transducer.
Since $(\Sigma\Gamma)^* \subseteq (\Sigma\Gamma)^*(\Sigma^* + \Gamma^*) \in \fsl$, it is now easy to see that

\begin{corollary}\label{thm:undec-automatic-sync}
It is undecidable whether a given rational relation has a uniformization by a synchronous subsequential transducer, that is, a $(\Sigma\Gamma)^*(\Sigma^* + \Gamma^*)$-controlled subsequential transducer.
\end{corollary}

\begin{corollary}\label{thm:undec-automatic-fsl}
It is undecidable whether a given rational relation has a uniformization by a subsequential transducer with finite shiftlag.
\end{corollary}

In the remainder of this section, we show that for a given automatic relation and a given regular synchronization language it is undecidable whether the relation has a subsequential uniformization according to the synchronization language.

First, we show a stronger result.
Recall that a regular synchronization language $S \subseteq (\SigmaIO)^*$ is a fine-grained way to describe the possible behavior of a transducer.
A little less fine-grained is a so-called regular \emph{control language}, that is a regular language $L \subseteq \{1,2\}$.
Such a language describes how a transducer may process input and output symbols, but does not specify which input and output symbols occur, e.g., $12211$ means that a transducer reads one input symbol, subsequently produces two output symbols, and then reads two input symbols. This notion of control language is the one used by \cite{conf/stacs/FigueiraL14}. Each control language $L$ naturally corresponds to a synchronization language by substituting $1$ with $\SigmaI$ and $2$ with $\SigmaO$, and to each synchronization language we can associate the control language obtained by the projection that replaces symbols from $\SigmaI$ with $1$, and symbols from $\SigmaO$ with $2$. 
So the definitions and notations introduced in \cref{sec:prelims} regarding synchronization languages can be applied in a natural way to control languages (in fact, these notions have originally been defined for control languages by \cite{conf/stacs/FigueiraL14}).

We show that for a given automatic relation and a given control language it is undecidable whether the relation has uniformization according to the control language.

\begin{theorem}\label{thm:undec-aut-control}
It is undecidable for a given automatic relation whether it has a uniformization by an $L$-controlled subsequential transducer for a given regular control language $L \subseteq \{1,2\}^*$.
\end{theorem}

\begin{proof}
 We give a reduction from the halting problem for Turing machines (TMs).
 Given a TM $M$ with a right-infinite tape, our goal is to describe an automatic relation $R_M$ and a regular control language $L_M$ such that $R_M$ can only be uniformized by an $L_M$-controlled subsequential transducer if $M$ does not halt for the empty input.
 
 Let $Q_M$ denote the state set of $M$, $q_0^M$ denotes the initial state of $M$, $\Gamma_M$ denotes the tape alphabet of $M$ including the blank symbol $\blank$, $\delta_M$ denotes the transition function of $M$, and $\vdash$ denotes the tape delimiter.
 We represent a configuration $c$ of $M$ in the usual way as $\vdash\! a_1\cdots a_kqb_1\cdots b_k$, where $a_1,\dots,a_k,b_1,\dots,b_\ell \in \Gamma_{M}$, $a_1\cdots a_kb_1\cdots b_\ell$ is the content of the tape of $M$, $q \in Q_M$ is the current control state of $M$, and the head of $M$ is on $b_1$.
 We assume that a configuration is represented without additional blanks, meaning a configuration representation does not start nor end with a blank.
 Let $\mathcal C^M$ denote the set of all configurations of $M$.
 For a configuration $c \in \mathcal C^M$ let $\mathit{succ}(c)$ denote its successor configuration.
 Without loss of generality, we assume that $\mathit{succ}(c)$ is at least as long as $c$, we call this non-erasing.
 Every TM can be made non-erasing by introducing a new tape symbol $B$ that is written instead of $\blank$ and the TM behaves on $B$ the same as on $\blank$.
 Furthermore, we call a configuration $c \in \mathcal C^M$ whose last letter is some $q \in Q_M$ and there is a transition of the form $\delta_M(q,\blank) = (q',a,R)$ extending configuration, because $|\mathit{succ}(c)| = |c| + 1$.
 We denote by $\mathcal C^M_E$ the set of extending configurations.

Let $\Sigma = Q_M \cupcdot \Gamma_M \cupcdot \{\$,\#\}$.
We define an automatic relation $R_M \subseteq \Sigma^* \times \Sigma^*$ whose domain consists of words of the form $\$c_1\$\cdots \$c_m\#$ where $c_1,\dots,c_{m-1}$ are configurations of $M$ and $c_m$ is a (prefix of a) configuration of $M$.
The idea is that we want to force a uniformizer to produce an output sequence that begins with $\$\vdash\!q_0^M\$\mathit{succ}(c_1)\$\dots \$\mathit{succ}(c_{m-1})$ for such an input.
The control language is defined in such a way that this is indeed possible (the subsequential transducer can produce $\mathit{succ}(c_i)$ while reading $c_i$).
And this should define a uniformization iff there is no halting computation of $M$ for the empty input.
That is, we have to define $R_M$ in such a way that the pair $(u,v) \in R_M$, where $u = \$c_1\$\cdots \$c_m\#$ and $v$ starts with $\$\vdash\!q_0^M\$\mathit{succ}(c_1)\$\cdots \$\mathit{succ}(c_{m-1})$ iff $c_1\$\dots \$c_{m}$ is not a halting computation of~$M$.

We describe a function $f^M_\mathit{succ}$ defined over all four-letter infixes of words from the domain of $R_M$ which defines how a subsequential transducer should build successor configurations.
First, given a configuration $c \in \mathcal C^M$, note that in order to output $\mathit{succ}(c)[1]\dots\mathit{succ}(c)[j]$, it suffices to have read $c[1]\cdots c[j+2]$. 
Furthermore, in order to determine the letter $\mathit{succ}(c)[j]$, it suffices to remember $c[j-1]c[j]c[j+1]c[j+2]$.
For example, consider the configuration $c$ of the form $a_1a_2qb_1b_2$ and the transition $\delta_M(q,b_1) = (q',c,L)$, then the successor configuration $\mathit{succ}(c)$ is $a_1q'a_2cb_2$.
To determine $\mathit{succ}(c)[2]$, $c[3]c[4]$ is needed.
To determine $\mathit{succ}(c)[3]$, $c[2]$ is needed, because the applied transition moves the head to the left. 

The function $f^M_\mathit{succ}$ extends this idea to infixes from words of the domain of $R_M$.
Formally, the type of the partial function $f^M_\mathit{succ}$ is $\Sigma\Sigma\Sigma\Sigma \to \Sigma \cup \Sigma\Sigma$.
We omit a formal definition of the mapping, instead, we provide some examples:
\begin{itemize}
  \item $a_1a_2qb_1$ maps to $q'$ if $\delta_M(q,b_1) = (q',c,L)$
  \item $a_2qb_1b_2$ maps to $c$ if $\delta_M(q,b_1) = (q',c,L)$
  \item $b_2b_3\$a_1$ maps to $b_3$ \quad\emph{configuration suffix is connected via $\$$ with configuration prefix}
  \item $b_3\$a_1q$ maps to $\$$
  \item $a_1a_2q\$$ maps to $a_2c$ if $\delta_M(q,\blank) = (q',c,R)$ \quad\emph{suffix of an extending configuration and $\$$ is seen, this is the only case the function has a two-letter target}
  \item $a_2q\$a_1$ maps to $q'$ if $\delta_M(q,\blank) = (q',c,R)$
\end{itemize}
To ensure that $f^M_\mathit{succ}$ is totally defined, we define $\mathit{succ}(c) = c$ for halting configurations $c$.

We design the relation $R_M$ and the control language $L_M$ such that a potential uniformizer is forced to behave according to $f^M_\mathit{succ}$. 
To ensure this, we introduce two constraints (formally below).
The output must end according to $f^M_\mathit{succ}$ applied to the suffix of the input, and we have to enforce that the potential uniformizer produces two letters only if the delimiter $\$$ after an extending configuration is seen.
The latter condition can not directly be enforced by $R_M$ nor $L_M$, but in $R_M$ it suffices to compare the length of the last configuration in the input and the last configuration of the output modulo two to detect deviations (explained in the correctness proof).

Before we define the relation, we introduce one shorthand notation.
Given $w \in \Sigma^*$, let $\mathit{last}_i(w)$ denote the letter $w[|w|-i]$ for $i \in \mathbbm N$, that is, the $i$th to last letter of $w$, and let $\mathit{last}_{i:j}(w)$ denote $w[|w|-i] \cdots w[|w|-j]$ for $i,j \in \mathbbm N$ with $i \geq j$.

Formally, the relation $R_M \subseteq \Sigma^* \times \Sigma^*$ contains a pair $(u,v)$ iff
\begin{enumerate}
  \item\label{i:1} $u = \$c_1\$\dots \$c_m\#$, where $c_1,\dots,c_{m-1} \in \mathcal C^M$, $c_m \in \prefs{\mathcal C^M}$, $m > 1$,
  \item\label{i:2} $v = \$c'_1\$\dots \$c'_{\ell}\#$, where $c'_1,\dots,c'_{\ell-1}\in \mathcal C^M$, $c_\ell \in \prefs{\mathcal C^M}$, $\ell > 1$,
  \item\label{i:3} $c'_1 :=\ \vdash\!q_0^M$ is the initial configuration of $M$ on the empty tape, 
  \item\label{i:4} $c_1\$\dots \$c_{m-1} \neq c'_1\$\dots \$c'_{\ell-2}$ or $c'_{\ell-1}$ is not a halting configuration of $M$,
  \item\label{i:5} $f^M_\mathit{succ}(\mathit{last}_{4:1}(u)) = 
    \begin{cases}
      \mathit{last}_{1}(v) & \text{if } f^M_\mathit{succ}(\mathit{last}_{4:1}(u)) \in \Sigma\\
      \mathit{last}_{2:1}(v) & \text{if } f^M_\mathit{succ}(\mathit{last}_{4:1}(u)) \in \Sigma\Sigma,
    \end{cases}$
  \end{enumerate}

  \emph{\cref{i:6,i:7} are designed under the assumption that if $|c_m| \leq 1$, then a uniformizer has not finished to produce $\mathit{succ(c_{m-1})}$.
  If $|c_m| = 2$, then $\$$ -- the delimiter after $\mathit{succ(c_{m-1})}$ -- is produced, and if $|c_m| \geq 3$, then the production of $\mathit{succ(c_{m})}$ has started.}
  \begin{enumerate} \setcounter{enumi}{5}
  \item\label{i:6} \begin{itemize}
    \item if $|c_m| > 3$, then $|c'_\ell| > 1$, \quad\emph{$c'_\ell$ refers to a prefix of $\mathit{succ(c_{m})}$ of length at least $2$}
    \item if $|c_m| = 3$, then $|c'_\ell| = 1$, \quad\emph{$c'_\ell$ refers to the first letter of $\mathit{succ(c_{m})}$}
    \item if $|c_m| = 2$, then $|c'_\ell| = 0$, \quad\emph{delimiter $\$$ after $\mathit{succ(c_{m-1})}$ is produced, thus, $c'_\ell = \varepsilon$}
    \item if $|c_m| \leq 1$, then $|c'_\ell| > 1$, \quad\emph{$c'_\ell$ refers to $\mathit{succ(c_{m-1})}$ which has length at least $2$}
    \end{itemize}
  \item\label{i:7} \begin{itemize}
    \item  if $|c_m| \geq 3$, then $(|c_m| - |c'_\ell|)\!\mod{2} = 0$,
    
    \item if $|c_m| = 2$, then
  $(|c_{m-1}| - |c'_{\ell-1}|)\!\mod{2} = 
  \begin{cases}
    1 & \text{if } c_{m-1} \in \mathcal C^M_E\\
    0 & \text{otherwise},
  \end{cases}$

    \item if $|c_m| \leq 1$, then
  $(|c_{m-1}| - |c'_{\ell}|)\!\mod{2} = 
  \begin{cases}
    1 & \text{if } c_{m-1} \in \mathcal C^M_E\\
    0 & \text{otherwise}.
  \end{cases}$
  \end{itemize}
\end{enumerate}

In combination with the control language $L_M$ defined below, \cref{i:5,i:6,i:7} ensure that an $L_M$-controlled subsequential uniformizer of $R_M$ has to output $\mathit{succ}(c_i)$ while reading $c_i$ according to $f^M_\mathit{succ}$ (as will be explained later).
It is not hard to see that $R_M$ is an automatic relation. For \cref{i:7} note that $(x-y)\mod{2} =\mathit{sign}((x\!\mod{2})-(y\!\mod{2}))$.
Replacing \cref{i:5,i:6,i:7} by $c'_{i+1} = \mathit{succ}(c_i)$ would make the relation non-automatic, in general.

The control language $L_M$ is defined by the regular expression $1222211(1(2+22))^*12$, we explain its use further below.

We claim that $R_M$ can be uniformized by an $L_M$-controlled subsequential transducer iff $M$ does not halt on the empty tape.
 
If $M$ does not halt, then an $L_M$-controlled subsequential transducer that produces
\[
 \$\underbrace{\vdash\!q_0^M}_{c_1'}\$f^M_\mathit{succ}(u[1:4])f^M_\mathit{succ}(u[2:5])\cdots f^M_\mathit{succ}(\mathit{last}_{4:1}(u))\#
\]
for input $u$ of the form $\$c_1\$\dots \$c_m\#$ uniformizes $R_M$:
First, note that this output guarantees that $c'_{i+1} = \mathit{succ}(c_i)$ for all $i \leq m-1$, and $c_{\ell-1} = \mathit{succ}(c_{m-1})$.
We show that the conditions of $R_M$ are satisfied:
We begin with \cref{i:5,i:6,i:7}.

Clearly, $f^M_\mathit{succ}(\mathit{last}_{4:1}(u))$ has the correct value.

The next statements are easy consequences of the definition of $f^M_\mathit{succ}$:
If $|c_m| \leq 1$, then the delimiter $\$$ after $\mathit{succ}(c_{m-1})$ is not included in the output, and $c'_\ell$ refers to $\mathit{succ}(c_{m-1})$.
Every configuration has length at least 2, thus, $|c'_\ell| \geq 2$.
If $|c_m| = 2$, then the output ends with $\$\#$, i.e., $c'_\ell = \varepsilon$.
If $|c_m| = 3$, then $|c'_\ell| = 1$ and $c'_\ell$ is the first letter of $\mathit{succ}(c_m)$.
Lastly, if $|c_m| > 3$, then $|c'_\ell| > 1$ and $c'_\ell$ is a prefix of $\mathit{succ}(c_m)$.
Hence, it is easy to see that all length constraints regarding $c'_\ell$ are satisfied.

We turn to the modulo constraints.
After $c_{m-1}$, $\$$ follows, thus, $f^M_\mathit{succ}$ at the end of $c_{m-1}\$$ yields two letters iff $c_{m-1} \in \mathcal C^M_E$.
Hence, $\mathit{succ}(c_{m-1})$ is one letter longer than $c_{m-1}$ iff $c_{m-1} \in \mathcal C^M_E$.
We have explained in the above paragraph that if $|c_m| \leq 1$, then $c'_\ell$ refers to $\mathit{succ}(c_{m-1})$ and if $|c_m| = 2$, then $c'_{\ell-1}$ refers to $\mathit{succ}(c_{m-1})$.
Thus, the condition if $|c_m| \leq 1$, then $(|c_{m-1}| - |c'_\ell|)\!\mod{2} = 1$ iff $c_{m-1} \in \mathcal C^M_E$ is satisfied.
Also, the condition if $|c_m| = 2$, then $(|c_{m-1}| - |c'_{\ell-1}|)\!\mod{2} = 1$ iff $c_{m-1} \in \mathcal C^M_E$ is satisfied.
After $c_m$, no $\$$ follows, thus $f^M_\mathit{succ}$ applied to shifting four-letter windows of $c_m$ (with three letters known of $c_m$ when the first letter of $c'_\ell$ is produced) yields that $|c_m|$ and $|c'_\ell|$ are either both even or both odd.
Thus, the condition if $|c_m| \geq 3$, then $(|c_m| - |c'_\ell|)\!\mod{2} = 0$ is satisfied. 

We turn to \cref{i:1,i:2,i:3,i:4}.
If $c_1$ is not equal to the initial configuration $q_0^M$, then clearly $c_1\$\dots \$c_{m-1} \neq c'_1\$\dots \$c'_{\ell-2}$.
Thus, assume $c_1$ is the initial configuration.
Then either there exists an $i \leq m-1$ such that $\mathit{succ}(c_i) \neq c_{i+1}$ or if no such $i$ exists then $\mathit{succ}(c_i)$ cannot be a halting configuration for all $i$ because $M$ does not halt.
In the former case we have that $c_{i+1} \neq \mathit{succ}(c_i) = c'_{i+1}$ which implies $c_1\$\dots \$c_{m-1} \neq c'_1\$\dots \$c'_{\ell-2}$.
In the latter case we have that $c'_{\ell-1} = \mathit{succ}(c_{m-1})$ is not a halting configuration.

We argue that it is possible for a subsequential transducer to realize
\[
  \$\vdash\!q_0^M\$f^M_\mathit{succ}(u[1:4])f^M_\mathit{succ}(u[2:5])\cdots f^M_\mathit{succ}(\mathit{last}_{4:1}(u))\#
\]
in an $L_M$-controlled fashion.
The idea is that the transducer, after producing $\$c_1'\$ := \$\vdash\!q_0^M\$$, behaves according to $f^M_\mathit{succ}$ shifting the considered four-letter window letter-by-letter.
Let us break down the intended use of $L_M$:
 \[
   \underbrace{12222}_{\substack{\text{read first $\$$ and}\\\text{output $\$\vdash\!q_0^M\$$}}}\hspace*{-1.5em}
  \overbrace{11}^{\substack{\text{lookahead for computing}\\\text{succ.\ config.\ correctly}}}\hspace*{-1.5em}
  \underbrace{(1(2+22))^*}_{\substack{\text{produce succ.\ config.\ with $(12)^*$}\\\text{use $122$ at \$ if tape}\\\text{content becomes longer}}}\hspace*{-0.5em}
  \overbrace{1}^{\text{final } \#}\hspace*{-1.5em}
  \phantom{\quad}\underbrace{2}_{\text{final } \#}
  \]
With this in mind it is easy to see that it is possible to realize the desired sequence in an $L_M$-controlled fashion such that all conditions of $R_M$ are satisfied.
  
Now assume that $M$ does halt for the empty input, and let $c_1,\ldots, c_m$ be the sequence of configurations corresponding to a halting computation.
Clearly, an $L_M$-controlled subsequential transducer that produces for the input $u$ of the form $\$c_1\$\cdots \$c_m\#$ the output
\[
  \$\vdash\!q_0^M\$f^M_\mathit{succ}(u[1:4])f^M_\mathit{succ}(u[2:5])\cdots f^M_\mathit{succ}(\mathit{last}_{4:1}(u))\#
\]
is not a uniformizer: $c_1\$\cdots\$c_{m-1} = c'_1\$\cdots\$c'_{\ell-2}$, and $c_{\ell-1}$ is a halting configuration of $M$.

We now show that any $L_M$-controlled subsequential uniformizer of $R_M$ behaves according to $f^M_\mathit{succ}$ (after producing $\$c_1'\$ := \$\vdash\!q_0^M\$$) which implies that if $R_M$ is uniformizable by an $L_M$-controlled subsequential transducer then $M$ does not halt on the empty input.
Assume such a uniformizer has so far behaved according to $f^M_\mathit{succ}$, and now produces something different.
Let $u'$ denote the current input sequence.
We distinguish two cases.

First, assume that the uniformizer in its last computation step has produced output $o$ such that $o \neq f^M_\mathit{succ}(\mathit{last}_{3:0}(u'))$, but $|o| = |f^M_\mathit{succ}(\mathit{last}_{3:0}(u'))|$.
Then, the next and last input symbol is $\#$.
The whole input sequence is $u = u'\#$.
Clearly, the condition $f^M_\mathit{succ}(\mathit{last}_{4:1}(u))$ of $R_M$ is violated.

Secondly, assume that the uniformizer in its last computation step has produced output $o$ such that $|o| \neq |f^M_\mathit{succ}(\mathit{last}_{3:0}(u'))|$.
As in the previous case, the next and last input symbol is $\#$.
The whole input sequence is $u = u'\#$.
The uniformizer must make one last single-letter output to be $L_M$-controlled, this output must be $\#$, otherwise the output is clearly not according to $R_M$.

To start, we assume that $\mathit{last}_{1}(u) = \$$.
Then $c_m = \varepsilon$.
Furthermore, $|f^M_\mathit{succ}(\mathit{last}_{4:1}(u))| = 2$ iff $c_{m-1} \in \mathcal C^M_E$.
The relation $R_M$ requires that $|c'_\ell| > 1$ if $|c_m| < 1$, hence, the output $o$ of the uniformizer does not contain $\$$.
Since $o$ was the first output not according to $f^M_\mathit{succ}$, we obtain that $|c_{m-1}| - |c'_{\ell}|$ is even iff $c_{m-1} \in \mathcal C^M_E$.
This violates the constraint if $|c_m| \leq 3$, then $(|c_{m-1}| - |c'_{\ell}|)\!\mod{2} = 1$ iff $c_{m-1} \in \mathcal C^M_E$.

We now assume that $\mathit{last}_{1}(u) \neq \$$.
Then $|f^M_\mathit{succ}(\mathit{last}_{4:1}(u))| = 1$, and $|o| = 2$, because $L_M$ ensures that an output is either one or two letters.
Depending on the length of $c_m$, different conditions of $R_M$ must be satisfied.

Assume that $|c_m| > 3$, then $o$ can not include $\$$, because that would start a new configuration prefix $c'_\ell$ which violates the length constraints of $c'_\ell$.
Thus, since $|o| = 2$, we obtain that $c_m$ is even iff $c'_\ell$ is odd, which violates that $(|c_m| - |c'_\ell|)\!\mod{2} = 0$ if $|c_m| \geq 3$.

Assume that $|c_m| = 3$.
This implies that the output produced before $o$ was $\$$, otherwise the subsequential transducer would have not behaved according to $f^M_\mathit{succ}$ in the computation step before $o$.
Hence, $o$ can not be of the form $\$\$$ or begin with $\$$, because $\$\$$ is not part of a valid output according to $R_M$.
If $o$ ends with $\$$, then $|c'_\ell| = 0$ which violates the length constraint.
Thus, $o$ contains no $\$$, meaning that $|c'_\ell| = 2$ which violates the length constraint.

Assume that $|c_m| = 2$.
The uniformizer must ensure that $|c'_\ell| = 0$, the only way to achieve this with an output of length two is either $\$\$$ or $o$ ends with $\$$ (and begins with some other letter).
However, $\$\$$ is not part of a valid output sequence.
In the other case, since all previous outputs were according to $f^M_\mathit{succ}$, we obtain that $(|c_{m-1}| - |c'_{\ell-1}|)\!\mod{2} = 0$ iff $c_{m-1} \in \mathcal C^M_E$, which violates that $(|c_{m-1}| - |c'_{\ell-1}|)\!\mod{2} = 1$ iff $c_{m-1} \in \mathcal C^M_E$.

Assume that $|c_m| \leq 1$.
The uniformizer must ensure that $|c'_\ell| > 1$, thus, $o$ can not contain~$\$$.
Thus, since $|o| = 2$ and all previous outputs were according to $f^M_\mathit{succ}$, we obtain that $(|c_{m-1}| - |c'_{\ell}|)\!\mod{2} = 0$ iff $c_{m-1} \in \mathcal C^M_E$, which violates that $(|c_{m-1}| - |c'_{\ell}|)\!\mod{2} = 1$ iff $c_{m-1} \in \mathcal C^M_E$.

We have proven that every $L_M$-controlled uniformizer of $R_M$ behaves according to $f^M_\mathit{succ}$ after the output $\$\vdash\!q_0^M\$$.
Together with $(u,v) \in R_M$ iff $M$ halts on the empty input, where
\[
 u = \$c_1\$\cdots\$c_m\# \in \mathrm{dom}(R_M), \text{ and}
\]
\[
 v = \$\vdash\!q_0^M\$f^M_\mathit{succ}(u[1:4])f^M_\mathit{succ}(u[2:5])\cdots f^M_\mathit{succ}(\mathit{last}_{4:1}(u))\#
\]
we obtain that $R_M$ is uniformizable by an $L_M$-controlled subsequential transducer iff $M$ halts on the empty input.
\end{proof}

We have seen that the problem whether an automatic relation has a uniformizer that behaves according to a given control language is undecidable, thus, it is also undecidable if we specify a synchronization language.

\begin{corollary}\label{cor:undec-aut-sync}
  It is undecidable whether a given automatic relation has a uniformization by an $S$-controlled subsequential transducer for a given regular synchronization language $S \subseteq (\SigmaIO)^*$.
\end{corollary}

\begin{proof}
 Consider the automatic relation $R_M$, based on a Turing machine $M$, and $L_M \subseteq \{1,2\}^*$ as in the proof of \cref{thm:undec-aut-control}.
 Let $S = \{ w \in (\SigmaIO)^* \mid w \text{ is $L_M$-controlled}\}$.
 It follows directly that it is undecidable whether $R_M$ has an $S$-controlled uniformization.
\end{proof}

%% file: sections/definability-problems.tex
In the previous section we considered uniformization problems in the context of synchronization languages. We now turn to the definability problem, so the question whether a given relation can be defined inside a given target language of synchronizations.

More formally, we are interested in the following problem.

\begin{definition}[Resynchronized definability problem]\label{def:definability}
  The \emph{resynchronized definability problem} asks given a regular source resp.\ target language $S$ resp.\ $T \subseteq (\SigmaIO)^*$ whether $\llbracket S \rrbracket \in \textnormal{\textsc{Rel}}(T)$.
\end{definition}
We say that $T' \subseteq T$ is a definition of $\llbracket S \rrbracket$ in $T$ if $T'$ is regular and $\llbracket S \rrbracket = \llbracket T' \rrbracket$, so if $T'$ is a witness for $\llbracket S \rrbracket \in \textnormal{\textsc{Rel}}(T)$.
We assume all source and target languages to be regular without explicitly mentioning it in the remainder.


We start by giving an example.

\begin{example} \label{ex:definability}
Let $\Sigma = \{a\}$ and $\Gamma = \{b\}$.
Consider
\[
  S = (ab)^* + (ab)^*(aa^+ + bb^+), \text{ and}
\]
\[
  T = a^*b^* + \underbrace{(ab)^*(aa^+ + bb^+)}_{U}.
\]
The relation $\llbracket S \rrbracket$ contains pairs of words where the words have the same length, or the difference between the length of the words is at least two. The target synchronization language contains all pairs of words with difference at least two in the synchronous encoding (denoted by the set $U$). The difference with $S$ is that the pairs of words of same length are contained in $T$ in the completely asynchronous encoding $a^*b^*$. Clearly, this set of words $a^nb^n$ for all $n$ is not a regular subset of $a^*b^*$. However, in order to define $\llbracket S \rrbracket$ in $T$, we can also select a larger subset of $a^*b^*$ that also contains pairs of words that are already covered by $U$. We illustrate below that this is indeed possible, that is, $\llbracket S \rrbracket \in \textnormal{\textsc{Rel}}(T)$.
We define the set 
\[
  T' = \underbrace{\{ uv \in a^*b^* \mid u \in a^*, v\in b^*, (|u|-|v|)\ \mathrm{mod}\ 2 = 0\}}_{M \subseteq a^*b^*}\ \cup\ U.
\]
Clearly, $T'$ is regular and $T' \subseteq T$.
The relation $\llbracket M \rrbracket$ captures all pairs of words where the difference between the length of the words is even which implies that $\llbracket (ab)^* \rrbracket \subseteq \llbracket M \rrbracket$.
Note that $\llbracket M \rrbracket \cap \llbracket U \rrbracket$ contains pairs of words where the difference between the length of the words is at least two and even.
Moreover, $(\llbracket M \rrbracket \setminus \llbracket (ab)^* \rrbracket) \subseteq \llbracket U \rrbracket$.
Consequently, $\llbracket T' \rrbracket = \llbracket S \rrbracket$.

\end{example}

\subsection{Overview and simple results}\label{sec:overview-def}

\begin{table}[t]
\begin{center}
\begin{tabular}{|l|c|c|c|} \hline
 \backslashbox[50mm]{target}{source} & \parbox{18mm}{rational\\$S \in \all$} &  \parbox{21mm}{automatic\\$S \in \fsl$} & \parbox{21mm}{recognizable\\$S \in \fs$} \\  \hline \hline

 \parbox{40mm}{\vspace{1mm}unrestricted\\$T = (\SigmaIO)^* \in \all$\vspace{1mm}} &{\color{gray}always} & {\color{gray}always} & {\color{gray}always}\\  \hline

 \parbox{45mm}{\vspace{1mm}synchronous\\$T\! =\! (\SigmaI\SigmaO)^*(\SigmaI^*\! +\! \SigmaO^*)\! \in\! \fsl$\vspace{1mm}} 	& \parbox{20mm}{\vspace{1mm}\centering undec.\\{\centering \cite{DBLP:journals/tcs/FrougnyS93}}\vspace{1mm}} & {\color{gray}always} & {\color{gray}always}\\  \hline

 \parbox{40mm}{\vspace{1mm}input before output\\$T = \SigmaI^*\SigmaO^*\in \fs$\vspace{1mm}} & \parbox{20mm}{\vspace{1mm}\centering undec.\\{\centering \cite{berstel2009}}\vspace{1mm}} & dec. \cite{DBLP:journals/ita/CartonCG06}  & {\color{gray}always}\\  \hline \hline

 $T \in \all$ 		& \parbox{20mm}{\vspace{1mm}\centering undec.\\{\centering (\cabref{prop:toRAT})}\vspace{1mm}} & \parbox{20mm}{\vspace{1mm}\centering undec.\\{\centering (\cabref{prop:toRAT})}\vspace{1mm}} & \parbox{20mm}{\vspace{1mm}\centering undec.\\{\centering (\cabref{prop:toRAT})}\vspace{1mm}}  \\  \hline

 $T \in \fsl$ 		& undec. &   \parbox{30mm}{\vspace{1mm}\centering open\\{\centering \small (dec.\ if target is unambiguous~(\cabref{prop:injectiveTarget}))}\vspace{1mm}}  	   & \parbox{22mm}{\vspace{1mm}\centering dec.~{\centering (\cabref{prop:RECtoAUT})}\\{\centering\small ($\llbracket S \rrbracket \in \textnormal{\textsc{Rel}}(T)$ iff~$\llbracket S \rrbracket \subseteq \llbracket T \rrbracket$)}\vspace{1mm}}\\  \hline

 $T \in \fs$ 		& undec. & \parbox{20mm}{\vspace{1mm}\centering dec.\\{\centering (\cabref{prop:AUTtoREC})}\vspace{1mm}}  & \parbox{22mm}{\vspace{1mm}\centering dec.~(\cabref{prop:RECtoAUT})\\{\centering\small ($\llbracket S \rrbracket \in \textnormal{\textsc{Rel}}(T)$ iff $\llbracket S \rrbracket \subseteq \llbracket T \rrbracket$)}\vspace{1mm}}\\  \hline \hline

 $\exists T \in \all$ & \multicolumn{3}{l|}{\color{gray}reduces to $T = (\SigmaIO)^*$}\\ \hline

 $\exists T \in \fsl$ 	& \multicolumn{3}{l|}{\color{gray}reduces to $T = (\SigmaI\SigmaO)^*(\SigmaI^* + \SigmaO^*)$}\\ \hline

 $\exists T \in \fs$ & \multicolumn{3}{l|}{\color{gray}reduces to $T = \SigmaI^*\SigmaO^*$}\\ \hline
\end{tabular}
\end{center}
\caption[]
{
Overview of results for instances of the resynchronized definability problem (upper and middle rows) and variants (lower rows).
The results are described in \cref{sec:overview-def}.
}
\label{tab:overviewdefinabilty}
\end{table}

In \cref{tab:overviewdefinabilty} we give an overview of known and new results. The table is organized in the same way as \cref{tab:overview} for the uniformization problems. As classes for source relations we consider rational, automatic, and recognizable relations. The upper three rows correspond to the question whether a relation is rational, automatic, and recognizable, respectively, because the target languages are effective canonical representations of their respective relation classes, see \cref{thm:canonical}.
In the context of definability, the lower three rows correspond to the same problems as in the upper three rows, respectively (while in the context of uniformization by subsequential transducers there are differences, see \cref{tab:overview}).

The three middle rows describe the instances of the resynchronized definability problem where a target synchronization language is given as part of the input. The undecidability and decidability results shown in the table are not very difficult to obtain, and are explained below.
The most interesting case, whether given automatic relation is definable inside a given target language of finite shiftlag, is considered in \cref{subsec:def-fsl}. The general case remains open, however, we solve some special cases of the problem.

Let us now turn to the undecidability results in the middle three rows. For the first column and the second and third middle row, undecidability follows directly for the choices of $T = (\SigmaI\SigmaO)^*(\SigmaI^* + \SigmaO^*)$ and $\SigmaI^*\SigmaO^*$, respectively.

The undecidability results for the first row in the middle part are a direct consequence of the undecidability of the universality problem for rational relations, see e.g., \cite{berstel2009}.

\begin{proposition}\label{prop:toRAT}
For every fixed source language $S$ with $\llbracket S \rrbracket = \SigmaI^* \times \SigmaO^*$ the resynchronized definability problem for given $T \in \all$ is undecidable.
\end{proposition}

\begin{proof}
If $\llbracket S \rrbracket = \SigmaI^* \times \SigmaO^*$, then the question if $S$ can be defined in $T$ is equivalent to asking whether $\llbracket T \rrbracket$ is universal. This corresponds to the universality problem for rational relations because there is a one-to-one correspondence between finite state transducers and regular synchronization languages (as explained in \cref{sec:prelims}).
\end{proof}



The decidability results in the middle rows of \cref{tab:overviewdefinabilty} are based on the following simple lemma.
\begin{lemma}\label{lemma:recognizable-subset}
Given a recognizable relation $R$ and $T \in \fs$, the set $\{ w \in T \mid \llbracket w \rrbracket \in R\}$ is regular.
\end{lemma}

\begin{proof}
Since $R$ is recognizable it can be effectively expressed as a finite union of products of regular languages, say $R = \bigcup_{i = 0}^n U_i \times V_i$ with regular languages $U_i \subseteq \SigmaI^*$ and $V_i \subseteq \SigmaO^*$.
An NFA for $\{ w \in S \mid \llbracket w \rrbracket \in R\}$ has to accept some $w \in (\SigmaIO)^*$ if $w \in T$, $\pi_\inp(w) \in U_i$ and $\pi_\outp(w) \in V_i$ for some $i$.
It is routine to construct such an NFA from NFAs for $T$ and the $U_i,V_i$.
\end{proof}


\begin{proposition}\label{prop:RECtoAUT}
The resynchronized definability problem is decidable if the given source language has finite shift and the given target language has finite shiftlag.
\end{proposition}

\begin{proof}
  Let $S \in \fs$ and $T \in \fsl$. We show that $\llbracket S \rrbracket \in \textnormal{\textsc{Rel}}(T)$ iff $\llbracket S \rrbracket \subseteq \llbracket T \rrbracket$. Clearly, if $\llbracket S \rrbracket \not\subseteq \llbracket T \rrbracket$, then $\llbracket S \rrbracket \not\in \textnormal{\textsc{Rel}}(T)$.
For the other direction, let $U := \{ w \in T \mid  \llbracket w \rrbracket \in \llbracket S \rrbracket\}$ be the set of all synchronizations in $T$ that synchronize a pair in $\llbracket S \rrbracket$. 
The relation $\llbracket S \rrbracket$ is recognizable, because $S$ is regular and has finite shift. Hence, $U$ is regular according to \cref{lemma:recognizable-subset}. If $\llbracket S \rrbracket \subseteq \llbracket T \rrbracket$, then each pair in $\llbracket S \rrbracket$ has at least one synchronization in $T$, and therefore $\llbracket U \rrbracket = \llbracket S \rrbracket$, which means that $\llbracket S \rrbracket \in \textnormal{\textsc{Rel}}(T)$.

Since $\llbracket S \rrbracket$ and $\llbracket T \rrbracket$ are automatic relations, the inclusion $\llbracket S \rrbracket \subseteq \llbracket T \rrbracket$ is decidable. 
\end{proof}

\begin{proposition}\label{prop:AUTtoREC}
The resynchronized definability problem is decidable if the given source language has finite shiftlag and the given target language has finite shift.
\end{proposition}

\begin{proof}
Given $S \in \fsl$ and $T \in \fs$.
The relation $\llbracket S \rrbracket$ is automatic, and for automatic relations it is decidable whether they are recognizable, see \cite{DBLP:journals/ita/CartonCG06}.
If $\llbracket S \rrbracket$ is not recognizable, then $\llbracket S \rrbracket \notin \textnormal{\textsc{Rel}}(T) \subseteq \Rec$.
If $\llbracket S \rrbracket$ is recognizable, it can be effectively represented as a recognizable relation, i.e., as some $S' \in \fs$.
Then, the problem reduces to the resynchronized definability problem for regular source and target languages with finite shift, which is decidable according to \cref{prop:RECtoAUT}.
\end{proof}

\subsection{Automatic source relations and targets with finite shiftlag}
\label{subsec:def-fsl}
We do not know whether the resynchronized definability problem for a given automatic relation and a given regular synchronization language with finite shiftlag is decidable. However, we provide some partial answers and show decidability for some special cases in \cref{sssec:minmax,sssec:injective,sssec:prefixrec}. We start by explaining how the problem is related to another open problem, namely the problem whether two given disjoint automatic relations $R_1$ and $R_2$ can be separated by a recognizable relation, that is, asking whether there exists a recognizable relation $R$ such that $R_1 \subseteq R$ and $R \cap R_2 = \emptyset$, see \cref{prop:sepability}.

This problem seems to lie on the border between decidability and undecidability, we give a brief, rather informal, explanation.
The following problem was recently shown to be undecidable by \cite{DBLP:conf/lics/Kopczynski16}: Given two visibly pushdown languages $L_1$ and $L_2$ that are disjoint, does there exist a regular language $L$ that separates $L_1$ and $L_2$, that is, $L_1 \subseteq L$ and $L_2 \cap L = \emptyset$?
Although this is an undecidable problem in general, we explain how to reduce the problem whether two disjoint automatic relations are separable by a recognizable relation to a more restricted variant.

An automatic relation can be translated into a so-called one-turn visibly pushdown language.
Think of a one-turn visibly pushdown language obtained from an automatic relation $R$ as a $\Sigma^*\Gamma^*$-controlled language of the form $\{ u^rv \mid (u,v) \text{ is in the relation } R\}$.
The idea is that in a visibly pushdown automaton, after reading $u^r$, the stack content is $u$, then while reading $v$ the stack is emptied and membership of $(u,v)$ in the relation is verified as follows:
Let $S$ be the $(\SigmaI\SigmaO)^*(\SigmaI^* \cup \SigmaO^*)$-controlled representation of $R$.
The visibly pushdown automaton repeats these steps: a letter from the input $u$ is read, and a letter from the stack content $v$ is read (and popped).
The pair of letters is used to simulate two steps in a DFA that recognizes the set $S$.

Turning to the separability problem, if there is a regular language separating two such one-turn visibly pushdown languages obtained from two disjoint automatic relations, then there is such a regular language that is $\Sigma^*\Gamma^*$-controlled (since all relevant words are in $\Sigma^*\Gamma^*$).
Consequently, the automatic relations that served as a starting point are separable by a recognizable relation of the form $\{ (u^r,v) \mid uv \text{ is in the separating language with }u\in\Sigma^* \text{ and } v\in \Gamma^*\}$.

\begin{proposition}\label{prop:sepability}
  The problem whether two automatic relations are separable by a recognizable relation can be reduced to the resynchronized definability problem for given source and target languages with finite shiftlag.
\end{proposition}

\begin{proof}
  Let $R_1 \subseteq \SigmaI^* \times \SigmaO^*$ and $R_2 \subseteq \SigmaI^* \times \SigmaO^*$ be two automatic relations with $R_1 \cap R_2 = \emptyset$.
  For the reduction, we define a source language $S$ and a target language $T$, both with finite shiftlag, such that
$\llbracket S \rrbracket \in \textnormal{\textsc{Rel}}(T)$ iff there is a recognizable relation that contains $R_1$ and has empty intersection with $R_2$.
  
  Let $S$ be the $(\SigmaI\SigmaO)^*(\SigmaI^* \cup \SigmaO^*)$-controlled representation of $\overline{R_2}$, that is, of the complement of $R_2$. 
  The set $S$ can be obtained starting from a $(\SigmaI\SigmaO)^*(\SigmaI^* \cup \SigmaO^*)$-controlled representation of $R_2$ via complementation and intersection.
  The target language $T$ is defined such that a definition of $\llbracket S \rrbracket$ as a subset of $T$ must contain all pairs from $R_1$ in the recognizable synchronization $\SigmaI^*\SigmaO^*$. The pairs that are neither in $R_1$ nor in $R_2$ can be contained using the recognizable or the automatic synchronization (or both). This is achieved by choosing $T := \SigmaI^*\SigmaO^* \cup M$, where $M$ is the  
  $(\SigmaI\SigmaO)^*(\SigmaI^*+\SigmaO^*)$-controlled representation of $\overline{R_1} \cap \overline{R_2}$.
  We show that $\llbracket S \rrbracket \in \textnormal{\textsc{Rel}}(T)$ iff $R_1$ and $R_2$ are separable by a recognizable relation.
  
  Assume $\llbracket S \rrbracket \in \textnormal{\textsc{Rel}}(T)$, and let $U \subseteq T$ be regular with $\llbracket U \rrbracket = \llbracket S \rrbracket = \overline{R_2}$. Let $U' = U \cap \SigmaI^*\SigmaO^*$. Then $\llbracket U'\rrbracket$ is a recognizable relation. Because  $\llbracket U'\rrbracket \subseteq \llbracket U \rrbracket = \overline{R_2}$, we obtain that $R_2 \cap \llbracket U'\rrbracket = \emptyset$.
  Further, $R_1 \subseteq \overline{R_2} = \llbracket U \rrbracket$. Since $\llbracket M \rrbracket \subseteq \overline{R_1}$, synchronizations in $M$ cannot contribute anything from $R_1$. From $U' = U \cap \SigmaI^*\SigmaO^* = U \setminus M$, we thus obtain that $R_1 \subseteq \llbracket U'\rrbracket$. Hence, $\llbracket U'\rrbracket$ is a recognizable relation that separates $R_1$ and $R_2$.
  
  Assume $R$ is a recognizable relation that separates $R_1$ and $R_2$.
  We have $\overline{R_2} = R \cup \llbracket M \rrbracket$, thus we can define $S$ in $T$ by $M \cup N$, where $N$ is the $\SigmaI^*\SigmaO^*$-controlled representation of~$R$. 
\end{proof}

\subsubsection{Special types of definitions}\label{sssec:minmax}
\input{sections/def-minmax}

\subsubsection{Unambiguous targets}\label{sssec:injective}
\input{sections/def-unambiguous}

\subsubsection{Prefix-recognizable targets}\label{sssec:prefixrec}
\input{sections/def-prefix}

%% file: sections/def-minmax.tex
The reduction in the proof of \cref{prop:sepability} suggests that the difficulty of the problem is coming from the fact that for certain pairs in the source relation, there are different synchronizations available in $T$. In the reduction, the pairs that are neither in $R_1$ nor in $R_2$ are contained in $T$ with their  $(\SigmaI\SigmaO)^*(\SigmaI^* \cup \SigmaO^*)$-controlled representation and their $\SigmaI^*\SigmaO^*$-controlled representation. These two representations are the two extreme cases in the sense that one representation is completely asynchronous (input before output), and the other representation is as synchronous as possible (alternating between input and output as long as possible). The definition of $\llbracket S \rrbracket$ in $T$ has to select some of the asynchronous representations but, in general, not all of them. This is also illustrated in \cref{ex:definability}.

In this section we show that this is indeed the main source of the difficulty of the problem. We define special types of definitions of $\llbracket S \rrbracket$ in $T$ that correspond to the selection of, roughly speaking, the most asynchronous representation of $\llbracket S \rrbracket$ in $T$ (called $\mathit{minsync}(S,T)$), the most synchronous representations of $\llbracket S \rrbracket$ in $T$ (called $\mathit{maxsync}(S,T)$), and the full representation of $\llbracket S \rrbracket$ in $T$ (called $\mathit{allsync}(S,T)$). We show that regularity of these representations of $\llbracket S \rrbracket$ in $T$ is decidable. In \cref{sssec:injective,sssec:prefixrec} we then give applications of these results.

For the formal definitions we need a measure of how ``(a)synchronous'' a word $w$ is in a synchronization language $T$ of finite shiftlag.
 A characterization for regular synchronization languages of finite shiftlag by \cite{conf/stacs/FigueiraL14} (formally stated in \cref{lemma:form} below) shows that such languages can be expressed as a finite union of concatenations of a regular language with finite lag and a regular language with finite shift.

Given $\gamma \in \mathbbm{N}$, we denote by $L_{\leq \gamma}$ the regular set of words over $\SigmaIO$ with $\leqlag{\gamma}$-lagged positions, i.e., $L_{\leq \gamma} = \{ u \in (\SigmaIO)^* \mid \mathit{lag}(u) \leq \gamma\}$.

\begin{lemma}[\cite{conf/stacs/FigueiraL14}]\label{lemma:form}
 Given a regular language $S \subseteq (\SigmaIO)^*$ with $\mathit{shiftlag}(S) < n$.
 It holds that $S \subseteq L_{\leq \gamma} \cdot (\SigmaI^*+\SigmaO^*)^n$ with $\gamma$ chosen as $2\left(n(|Q|+1)+1\right)$, where $Q$ is the state set of an NFA recognizing $S$.
\end{lemma}

Intuitively, the finite lag part of this representation of $T$ corresponds to the synchronous part.
So the idea is that we consider for $w \in T$ the longest prefix $x$ of $w$ such that $x^{-1}T$ is not of finite shift. The longer such a prefix, the ``more synchronous'' the word $w$ is inside $T$.
Formally, these ideas are captured by a partial order on words in $T$, as defined below. 

%

Given a synchronization language $T \in \fsl$, we define a partial order on $T$.
Given $w, w' \in T$ such that $\llbracket w \rrbracket = \llbracket w' \rrbracket$, we let $w \preceq_T w'$ if $(w'[1,i])^{-1}T \in \fs$  implies $(w[1,i])^{-1}T \in \fs$ for all $i$. As usual,
we let $w \approx_T w'$ if $w \preceq_T w'$ and $w' \preceq_T w$, and $w \prec_T w'$ if  $w \preceq_T w'$ and $w \not\approx_T w'$.
We omit the index $T$ if it is clear from the context.

Given two synchronization languages $S,T \in \fsl$, we define three sets of synchronizations, namely,
\begin{align*}
  & \mathit{minsync}(S,T) =  \\
  &\quad \{ w \in T \mid \llbracket w \rrbracket \in \llbracket S \rrbracket \text{ and } w \preceq_T w' \text{ for all } w' \in T \text{ such that } \llbracket w \rrbracket = \llbracket w' \rrbracket\}, \\
  & \mathit{maxsync}(S,T) =  \\
  &\quad \{ w \in T \mid \llbracket w \rrbracket \in \llbracket S \rrbracket \text{ and } w \succeq_T w' \text{ for all } w' \in T \text{ such that } \llbracket w \rrbracket = \llbracket w' \rrbracket\}, \text{ and} \\
  & \mathit{allsync}(S,T) = \{ w \in T \mid \llbracket w \rrbracket \in \llbracket S \rrbracket\}.
\end{align*}

If one of the above sets is regular, and $\llbracket S \rrbracket \subseteq \llbracket T \rrbracket$, then this gives a $T$-controlled definition of $\llbracket S \rrbracket$ because $\mathit{minsync}(S,T)$ and $\mathit{maxsync}(S,T)$ contain exactly one representation of every element in $\llbracket S \rrbracket$, and $\mathit{allsync}(S,T)$ contains all representations of $\llbracket S \rrbracket$-elements that are in $T$.
Our main goal is to show that it is decidable whether these sets are regular, see \cref{lemma:minmaxsync}. 
Towards that we introduce some additional notations, definitions, and lemmas.

\subparagraph{Additional notations.}

Given a DFA $\mathcal A$ over $\SigmaIO$, we denote by $Q_\mathcal A^\mathsf{FS} \subseteq Q_\mathcal A$ the set that contains a state $q$ if $L(\mathcal A_q) \in \fs$, where $\mathcal A_q$ denotes $\mathcal A$ with initial state $q$. The following lemma is a simple consequence of this definition.

\begin{lemma}\label{remark:FS}
Let $T \in \fsl$ and $\mathcal A$ be a DFA with $L(\mathcal A) = T$. For all $w \in (\SigmaIO)^*$
we have that $w^{-1}T \in \fs$ iff $\delta_\mathcal A^*(w) \in Q_\mathcal A^\mathsf{FS}$.
\end{lemma}

We say that two words $x,y \in (\SigmaIO)^*$ are \emph{compatible} if they are of same length and there exist $u,v \in (\SigmaIO)^*$ such that $\llbracket xu \rrbracket = \llbracket yv \rrbracket$. So two words of same length are compatible if they can be extended such that they synchronize the same pair. For example, with $\Sigma = \{a,b\}$ and $\Gamma = \{c,d\}$, the words $acab$ and $cadc$ are compatible because $\llbracket acab \cdot dc \rrbracket = \llbracket cadc \cdot ab \rrbracket = (aab,cdc)$.
Furthermore, we define $\mathit{diff}(x,y) = (u,v)$ if $x$ and $y$ are compatible and $u,v$ are the shortest words such that $\llbracket xu \rrbracket = \llbracket yv \rrbracket$, and let $|\mathit{diff}(x,y)| = |uv|$. Note that $u$ and $v$ are unique because $|x|=|y|$ implies that one of $u,v$ consists only of input symbols, and the other one only of output symbols. Taking the previous example, we have $\mathit{diff}(acab,cadc) = (dc,ab)$

As a tool in our proofs we use an automaton that synchronously reads pairs of words $x,y$ that have bounded lag and, if $x,y$ are compatible, it reaches the state $(u,v)$ for $\mathit{diff}(x,y) = (u,v)$. This is explained in more detail below.

First note that for compatible words of bounded lag, the length of their difference is bounded.
\begin{remark} \label{rem:compatible-bounded-lag}
Let $x,y$ be compatible with $\mathit{diff}(x,y) = (u,v)$. If $\mathit{lag}(x),\mathit{lag}(y) \leq k$, then $|u|,|v| \le k$.
\end{remark}

Given $k \in \mathbbm N$, we define the relation $\mathbb D_k \subseteq (\SigmaIO)^* \times (\SigmaIO)^* $ as 
\[
 \mathbb D_k := \{ (x,y) \mid \text{$x$ and $y$ are compatible and $\mathit{lag}(x),\mathit{lag}(y) \leq k$}\}.
\]
The relation $\mathbb D_k$ is automatic, that is, the set $\{ x \otimes y \mid (x,y) \in \mathbb D_k\}$ is regular, where $x \otimes y$ denotes the \emph{convolution} of $x$ and $y$ defined as the word $(x[1],y[1])\cdots(x[n],y[n]) \in (\Sigma \times \Gamma)^*$, where $n = |x| = |y|$ (if $|x| \neq |y|$ the convolution is not defined).
Due to \cref{rem:compatible-bounded-lag}, the set of convolutions of pairs in $\mathbb D_k$ is recognized by a DFA \label{page:dfa} $\mathcal D_{k}$ that keeps track of the difference of the words. Formally, $\mathcal D_{k} = (Q,(\SigmaIO) \times (\SigmaIO), (\varepsilon,\varepsilon),\Delta,Q)$, where $Q = \{ (w_0,w_1) \mid w_0,w_1 \in \SigmaI^* \cup \SigmaO^*, |w_0|=|w_1|, |w_0|,|w_1|\leq k\}$ and the transitions in $\Delta$ update the difference for the next input $(a_0,a_1)$, which can formally be defined by $\bigl( (w_0,w_1),(a_0,a_1),(y_0z_0,y_1z_1) \bigr) \in \Delta$ with 
  \[ y_i = 
    \begin{cases} 
      a_i^{-1}w_i & \text{if } a_i \in \SigmaI, w_i \in \SigmaI^+ \text{ or } a_i \in \SigmaO, w_i \in \SigmaO^+\\
      w_i & \text{if } a_i \in \SigmaI, w_i \in \SigmaO^+ \text{ or } a_i \in \SigmaO, w_i \in \SigmaI^+ \\
      \varepsilon & \text{if } w_i = \varepsilon, a_i = a_{1-i}\\
    a_{1-i} & \text{if } w_i = \varepsilon, a_i \neq a_{1-i}, \text{ and}
    \end{cases}
  \]
  \[ z_i = 
    \begin{cases} 
      \varepsilon & \text{if } a_{1-i} \in \SigmaI, w_{1-i} \in \SigmaI^+ \text{ or } a_{1-i} \in \SigmaO, w_{1-i} \in \SigmaO^+\\
      a_{1-i} & \text{if } a_{1-i} \in \SigmaI, w_{1-i} \in \SigmaO^+ \text{ or } a_{1-i} \in \SigmaO, w_{1-i} \in \SigmaI^+ \\
      \varepsilon & \text{if } w_i = \varepsilon, a_i = a_{1-i}\\
    a_{1-i} & \text{if } w_i = \varepsilon, a_i \neq a_{1-i}
    \end{cases}
  \]
  for $i = 0,1$.
  Note that if the operation $a_i^{-1}w_i$ does not yield an element from $(\SigmaIO)^*$ then there exists no transition from $(w_0,w_1)$ with label $(a_0,a_1)$, because this indicates incompatibility.
  
  By an easy but cumbersome induction one can show that the above definition of $\mathcal D_k$ is such that the automaton indeed computes the difference of two compatible string with lag at most $k$, as stated in the following lemma.

\begin{lemma}\label{lemma:kdelay}
Given $x,y$ compatible with $\mathit{lag}(x),\mathit{lag}(y) \le k$ and $\mathit{diff}(x,y) = (u,v)$, then $\mathcal D_{k}\colon (\varepsilon,\varepsilon) \xrightarrow{x \otimes y} (u,v)$.
\end{lemma}

Lastly, given $T \in \fsl$, we define the set $\fse(T)$ -- finite shift entry -- as
\[
\fse(T) := \{ x \in \mathit{Prefs}(T) \mid x^{-1}T \in \fs \text{ and } (x')^{-1}T \notin \fs \text{ for all $x'$ proper prefix of $x$} \},
\]
that is, the set of all prefixes of $T$ that induce an entry point into some finite shift remainder of $T$.
It is easy to see that this set is regular because by \cref{remark:FS} it is sufficient to accept those words that enter a state from $Q_\mathcal A^\mathsf{FS}$ in a DFA $\mathcal A$ for $T$ for the first time.
Furthermore, $\mathit{lag}(\fse(T)) \leq \gamma$, where $\gamma$ is chosen according to \cref{lemma:form}.

\subparagraph{Main lemmas.} 
We are ready to show our key lemmas. The decidability of regularity of the sets $ \mathit{minsync}(S,T)$, $\mathit{maxsync}(S,T)$, and $\mathit{allsync}(S,T)$ is stated in \cref{lemma:minmaxsync}.


We start with a helpful lemma about the regularity of sets of non-minimal synchronizations, and then consider the sets $ \mathit{minsync}(T,T)$, $\mathit{maxsync}(T,T)$ without the additional parameter $S$ in \cref{lemma:minreg} and \cref{lemma:if-maxreg}.

\begin{lemma}\label{lemma:larger-sync-reg}
For $T \in \fsl$ and a regular $T' \subseteq T$, the set $\{w \in T \mid\text{ there is } w' \in T' \text{ such that } w \succ w'\}$  is regular.
\end{lemma}
\begin{proof}
Let $U := \{w \in T \mid\text{ there is } w' \in T' \text{ such that } w \succ_T w'\}$.

A synchronization $w \in T$ is in $U$ if $w=xz$ such that $x^{-1}T \notin \fs$ and there exists some $y \in \fse(T)$ such that $x,y$ are compatible and $\llbracket w \rrbracket \in \llbracket y(y^{-1}T') \rrbracket$. The last condition ensures that $y$ can be extended to some $w'=yz' \in T'$ with $\llbracket w \rrbracket = \llbracket w' \rrbracket$.

This description allows us to represent $U$ as
\begin{align*}
   & \bigcup_{
    \mathclap{\substack{
    x \text{ s.t.\ } x^{-1}T \notin \fs\\
      y \in \fse(T)
    }}
  }\ 
  \{ xz \in T \mid \llbracket xz \rrbracket \in \llbracket y (y^{-1}T') \rrbracket\} 
  = \bigcup_{
    \mathclap{\substack{
    x \text{ s.t.\ } x^{-1}T \notin \fs\\
      y \in \fse(T)
    }}
  }\ 
  x \cdot \{ z \in x^{-1}T \mid \llbracket z \rrbracket \in \llbracket x \rrbracket^{-1} (\llbracket y \rrbracket \llbracket y^{-1}T' \rrbracket)\} \\
  = &   \bigcup_{
    \mathclap{\substack{
      x \text{ s.t.\ } x^{-1}T \notin \fs\\
        y \in \fse(T)
      }}
    }\ 
    x \cdot \{ z \in x^{-1}T \mid \llbracket z \rrbracket \in \llbracket v \rrbracket^{-1} (\llbracket u \rrbracket \llbracket y^{-1}T' \rrbracket)\}, \text{where $\mathit{diff}(x,y) = (u,v)$.}
\end{align*}

Using this representation, our goal is to show that $U$ is regular by rewriting it it terms of automata for $T,T'$.
Let $\mathcal A,\mathcal A'$ denote DFAs for $T,T'$ respectively. For simplicity, we assume that $\mathcal A$ and $\mathcal A'$ use the same transition structure, and only differ in their sets of final states (this can be achieved by taking the product of the two transition structures).
Recall that $x$ and $y$ (chosen as above) have a lag of at most $\gamma$, where $\gamma$ is chosen according to \cref{lemma:form}.
According to \cref{lemma:kdelay}, the run of $\mathcal D_{\gamma}$ on $x \otimes y$ ends in the state $\mathit{diff}(x,y)$, say $(u,v)$, which indicates that $\llbracket x u \rrbracket = \llbracket y v \rrbracket$.
So we can write the representation of $U$ as
\[
\bigcup_{
  \mathclap{\substack{
    p \in Q_\mathcal A \setminus Q_\mathcal A^{\mathsf{FS}} \\
     q \in Q_\mathcal A^{\mathsf{FS}}\\
    u,v \in \SigmaI^{\leq \gamma} \cup \SigmaO^{\leq \gamma}
  }}
  }\ 
  M_{p,q,u,v} \cdot 
  \overbrace{
  \{ z \in L(\mathcal A_p) \mid \llbracket z \rrbracket \in \underbrace{
    \llbracket u \rrbracket^{-1} \bigl(\llbracket v \rrbracket \llbracket L(\mathcal A_q') \rrbracket\bigr)
    }_{R_{q,u,v}}\}
  }^{L_{p,q,u,v}},
\]
where
\[
\begin{array}{ll}
M_{p,q,u,v} = \{ x \in (\SigmaIO)^* \mid & \text{there is } y \in \fse(T) \text{ such that } \mathcal A\colon q_0^\mathcal A \xrightarrow{x} p, \\
&
\mathcal A\colon q_0^\mathcal A \xrightarrow{y} q, \text{ and } \mathcal D_{\gamma}\colon (\varepsilon,\varepsilon) \xrightarrow{x \otimes y} (u,v)\}.
\end{array}
\]
The set $M_{p,q,u,v}$ is easily seen to be regular.
Since $L(\mathcal A_q)$ is regular and has finite shift, the relation $R_{q,u,v}$ is recognizable.
Hence, the set $L_{p,q,u,v}$ is regular according to \cref{lemma:recognizable-subset}.

In conclusion, the set $U$ is representable as a finite union of regular sets, thus, it is regular.
\end{proof}

\begin{lemma}\label{lemma:minreg}
For each $T \in \fsl$ the set $\mathit{minsync}(T,T)$ is regular.
\end{lemma}

\begin{proof}
The set $\mathit{minsync}(T,T)$ is equal to
$
  T \setminus \{ w \in T \mid \text{ there is } w' \in T \text{ such that } w \succ w'\},
$
that is, $T$ without all synchronizations that are not minimal. An application of \cref{lemma:larger-sync-reg} with $T' =T$ yields that the set of non-minimal synchronizations in $T$ is regular, and hence also its complement $\mathit{minsync}(T,T)$.
\end{proof}

In contrast to \cref{lemma:minreg}, given $T \in \fsl$, the set $\mathit{maxsync}(T,T)$ is not regular in general.

\begin{example}\label{ex:max-not-reg}
Consider $T = (\SigmaI\SigmaO)^* + \SigmaI^*\SigmaO^*$.
The set $\mathit{maxsync}(T,T)$ is of the form 
\[
  T \setminus \{xy \in \SigmaI^*\SigmaO^* \mid x \in \Sigma^*, y \in \SigmaO^*, |x| = |y|\}.
\]
It is easy to see that the set to be removed is not regular, thus, $\mathit{maxsync}(T,T)$ is not regular.
\end{example}

In fact, regularity of $\mathit{maxsync}(T,T)$ turns out to be a strong property because then the resynchronized definability problem reduces to the question of regularity of $\mathit{maxsync}(S,T)$, as stated in the following lemma.
\begin{lemma}\label{lemma:if-maxreg}
Let $S$ and $T \in \fsl$ such that $\mathit{maxsync}(T,T)$ is regular.
Then $\llbracket S \rrbracket \in \textnormal{\textsc{Rel}}(T)$ iff $\mathit{maxsync}(S,T)$ is regular and $\llbracket S \rrbracket \subseteq \llbracket T \rrbracket$.
\end{lemma}
\begin{proof}
If $\mathit{maxsync}(S,T)$ is regular and $\llbracket S \rrbracket \subseteq \llbracket T \rrbracket$, then $\llbracket S \rrbracket = \llbracket \mathit{maxsync}(S,T)\rrbracket \in \textnormal{\textsc{Rel}}(T)$.

Now assume that $\llbracket S \rrbracket \in \textnormal{\textsc{Rel}}(T)$, and let $T' \subseteq T$ be a regular set such that $\llbracket T' \rrbracket = \llbracket S \rrbracket$. Then $\mathit{maxsync}(S,T) = \mathit{maxsync}(T',T)$, and we have
\[
\mathit{maxsync}(T',T) = \mathit{maxsync}(T,T) \cap (T' \cup \underbrace{\{w \in T \mid \text{ there is } w' \in T' \text{ with } w \succ w'\}}_{U}).
\]
The set $U$ is regular according to \cref{lemma:larger-sync-reg}, and hence $\mathit{maxsync}(T',T)$ is a  Boolean combination of regular sets.
\end{proof}

We note that \cref{lemma:if-maxreg} fails without the regularity assumption:
If $\mathit{maxsync}(T,T)$ is not regular, we can choose $S=T$ and obtain that $\llbracket S \rrbracket  = \llbracket T \rrbracket \in \textnormal{\textsc{Rel}}(T)$, and $\mathit{maxsync}(S,T) = \mathit{maxsync}(T,T)$ is not regular (and \cref{ex:max-not-reg} shows that there are regular sets $T$ such that $\mathit{maxsync}(T,T)$ is not regular).

We now prove our main lemma on the decidability of the regularity of the different sets of synchronizations.

\begin{lemma}\label{lemma:minmaxsync}
Given two synchronization languages $S,T \in \fsl$.
The regularity of the followings sets is effectively decidable.
\begin{enumerate}
  \item $\mathit{minsync}(S,T)$,
  \item $\mathit{maxsync}(S,T)$, and
  \item $\mathit{allsync}(S,T)$.
\end{enumerate}
\end{lemma}

\begin{proof}
First note that the definition of the sets $\mathit{minsync}(S,T)$, $\mathit{maxsync}(S,T)$, and $\mathit{allsync}(S,T)$ does not depend on the precise representation of $S$, but only on the relation $\llbracket S \rrbracket$. Hence, we can choose a representation of $\llbracket S \rrbracket$ inside $\fsl$ that is most convenient for our purposes.

According to \cref{lemma:form}, $T \subseteq L_{\leq \gamma} \cdot (\SigmaI^* + \SigmaO^*)^m$ for some $\gamma$ and $m$. This means the synchronizations in $T$ have lag at most $\gamma$ before entering the finite shift part. We choose for $S$ what we call the full $\gamma$-lagged representation, which means that $S$ satisfies $S = \{w \in L_{\leq \gamma} \cdot (\SigmaI^* + \SigmaO^*) \mid \llbracket w \rrbracket \in \llbracket S \rrbracket\}$. We can obtain such a representation of $\llbracket S \rrbracket$ as follows: Let $S'$ be the canonical $(\SigmaI\SigmaO)^*(\SigmaI^* + \SigmaO^*)$-controlled representation of $\llbracket S \rrbracket$. Then $\{ y \mid \text{exists } x \in S' \text{ and } \mathcal{D}_\gamma: (\varepsilon,\varepsilon) \xrightarrow{x \otimes y}(\varepsilon,\varepsilon)\}$ is the full $\gamma$-lagged representation of $\llbracket S \rrbracket$, which is clearly regular.

So in the following, we assume that $S$ is the full $\gamma$-lagged representation of $\llbracket S \rrbracket$, and let $\mathcal A$ be a DFA for $S$.
The advantage of this representation is the following. Let $y \in (\SigmaIO)^*$ with $\mathit{lag}(y) \le \gamma$. Then $ \llbracket y \rrbracket^{-1} \llbracket S \rrbracket =  \llbracket y^{-1}S \rrbracket$, that is, if there is a pair in $\llbracket S \rrbracket$ that starts with $\llbracket y \rrbracket$, then there is a synchronization of that pair in $S$ that starts with $y$. 

Let $\mathcal B$ be a DFA for $T$.

\begin{claim}
It is decidable whether $\mathit{minsync}(S,T)$ is regular.
\end{claim}

\begin{proof}
Note that $\mathit{minsync}(S,T) = \mathit{allsync}(S,\mathit{minsync}(T,T))$.
\Cref{lemma:minreg} yields that $\mathit{minsync}(T,T)$ is regular.
Together with the next claim, decidability follows.
\end{proof}

In the proofs of the next claims we use that it is decidable whether an automatic relation is recognizable, see \cite{DBLP:journals/ita/CartonCG06}.

\begin{claim}\label{claim:allreg}
It is decidable whether $\mathit{allsync}(S,T)$ is regular.
\end{claim}
  
\begin{proof}
We show that $\mathit{allsync}(S,T)$ is not regular iff there exists some $y \in \fse(T)$ such that the relation $\llbracket y^{-1}S \rrbracket \cap \llbracket y^{-1}T \rrbracket$ is not recognizable. 

Assume that $\llbracket y^{-1}S \rrbracket \cap \llbracket y^{-1}T \rrbracket$ is not recognizable for some $y \in \fse(T)$.
Hence, $\llbracket y^{-1}\mathit{allsync}(S,T) \rrbracket = \llbracket y^{-1}S \rrbracket \cap \llbracket y^{-1}T \rrbracket$ is not recognizable.
Since $y^{-1}T \in \fs$, the set $y^{-1}\mathit{allsync}(S,T) \subseteq y^{-1}T$ has finite shift.
However, it is not regular, because every regular set with finite shift describes a recognizable relation.
Thus, $y^{-1}\mathit{allsync}(S,T)$ is not regular which implies that $\mathit{allsync}(S,T)$ is not regular.

Conversely, assume that $\llbracket y^{-1}S \rrbracket \cap \llbracket y^{-1}T \rrbracket$ is recognizable for all $y \in \fse(T)$.
We partition the set $\mathit{allsync}(S,T)$ into $\{ w \mid w^{-1}T \notin \fs\} \cap \mathit{allsync}(S,T)$ and $\{ w \mid w^{-1}T \in \fs\} \cap \mathit{allsync}(S,T)$ and show that both sets are regular.

The first set is regular, because if $w^{-1}T \notin \fs$, then $\mathit{lag}(w) \leq \gamma$, thus, we can rewrite first set as $\{ w\in S \cap T \mid w^{-1}T \notin \fs\}$.

The second set is described as
\begin{align*}
 & \bigcup_{\mathclap{\substack{
    y \in \fse(T)
  }}
  }\ 
  y \cdot (y^{-1}\mathit{allsync}(S,T))
  = \bigcup_{\mathclap{\substack{
    y \in \fse(T)
  }}
  }\ 
  y \cdot \{ z \in y^{-1}T \mid \llbracket z \rrbracket \in \llbracket y^{-1}S \rrbracket \cap \llbracket y^{-1}T \rrbracket\} \\
  = & \bigcup_{\mathclap{\substack{
    p \in Q_\mathcal A\\
    q \in Q_\mathcal B 
    }}
    }\ 
     M_{p,q} \cdot 
\{ z \in L(\mathcal B_q) \mid \llbracket z \rrbracket \in \llbracket L(\mathcal A_p)\rrbracket \cap \llbracket L(\mathcal B_q)\rrbracket\},
\end{align*}
where $M_{p,q} = \{ y \in \fse(T) \mid \mathcal A\colon q_0^\mathcal A \xrightarrow{y} p, \mathcal B\colon q_0^\mathcal B \xrightarrow{y} q \}$.
The set $M_{p,q}$ is clearly regular.
If $M_{p,q}$ is non-empty then $\llbracket L(\mathcal A_p)\rrbracket \cap \llbracket L(\mathcal B_q)\rrbracket$ is recognizable because then $\llbracket L(\mathcal A_p) \rrbracket = \llbracket y^{-1}S \rrbracket$, $\llbracket L(\mathcal B_q) \rrbracket = \llbracket y^{-1}T \rrbracket$ for some $y \in \fse(T)$, and $\llbracket y^{-1}S \rrbracket \cap \llbracket y^{-1}T \rrbracket$ is recognizable by assumption.
Thus, the set $\{ z \in L(\mathcal B_q) \mid \llbracket z \rrbracket \in \llbracket L(\mathcal A_p)\rrbracket \cap \llbracket L(\mathcal B_q)\rrbracket\}$ is regular according to \cref{lemma:recognizable-subset}.

We have shown $\mathit{allsync}(S,T)$ is regular iff $\llbracket y^{-1}S \rrbracket \cap \llbracket y^{-1}T \rrbracket$ is recognizable for all $y \in \fse(T)$.
With $p,q$ as above, $\llbracket y^{-1}S \rrbracket \cap \llbracket y^{-1}T \rrbracket = \llbracket L(\mathcal A_p)\rrbracket \cap \llbracket L(\mathcal B_q)\rrbracket$, and hence there are only finitely many such relations. Further, $\llbracket L(\mathcal A_p)\rrbracket \cap \llbracket L(\mathcal B_q)\rrbracket$ is an automatic relation because $\llbracket L(\mathcal A_p)\rrbracket$ and $\llbracket L(\mathcal B_q)\rrbracket$ are automatic. Hence, it is decidable whether $\llbracket L(\mathcal A_p)\rrbracket \cap \llbracket L(\mathcal B_q)\rrbracket$ is recognizable. 
\end{proof}

The proof of the next claim is similar to the proof of the previous claim but more involved.

\begin{claim}\label{claim:maxreg}
It is decidable whether $\mathit{maxsync}(S,T)$ is regular.
\end{claim}
  
\begin{proof}
Our goal is to show that $\mathit{maxsync}(S,T)$ is regular iff $\llbracket y^{-1} \mathit{maxsync}(S,T)\rrbracket$ is recognizable for all $y \in \fse(T)$. The proof of this also provides a way to decide this property.

Assume that $\mathit{maxsync}(S,T)$ is regular.
Consider some $y \in \fse(T)$.
Clearly, $y^{-1}\mathit{maxsync}(S,T)$ is regular.
Consequently, $\llbracket y^{-1}\mathit{maxsync}(S,T) \rrbracket$ is recognizable, because $y^{-1}\mathit{maxsync}(S,T) \subseteq y^{-1}T \in \fs$ is regular and has finite shift.

Conversely, assume that $\llbracket y^{-1} \mathit{maxsync}(S,T)\rrbracket$ is recognizable for all $y \in \fse(T)$, and consider some $y \in \fse(T)$. We describe $\llbracket y^{-1} \mathit{maxsync}(S,T) \rrbracket$ as an intermediate step to showing that $\mathit{maxsync}(S,T)$ is regular.
First, we have that $z \in y^{-1}(\mathit{maxsync}(S,T))$ if $z \in y^{-1}T$ and $\llbracket z \rrbracket \in \llbracket y^{-1}S \rrbracket$, and there is no $x$ compatible to $y$ such that $x^{-1}T \notin \fs$ and $\llbracket yz \rrbracket \in \llbracket x(x^{-1}T)\rrbracket$ because this would mean that $yz$ is not a maximal synchronization. We obtain that 
\begin{align*}
  \llbracket y^{-1} \mathit{maxsync}(S,T) \rrbracket =
  & \left(\llbracket  y^{-1}S  \rrbracket \cap \llbracket  y^{-1}T  \rrbracket\right) \setminus \Bigl(\bigcup_{\mathclap{\substack{
      x \text{ compatible to } y\\
      x^{-1}T \notin \fs
    }}
  }\ 
  \llbracket y \rrbracket^{-1} \llbracket x(x^{-1}T) \rrbracket
  \Bigr)\\
    = & \left(\llbracket  y^{-1}S  \rrbracket \cap \llbracket  y^{-1}T  \rrbracket\right) \setminus \Bigl(\bigcup_{\mathclap{\substack{
    x \text{ compatible to } y\\
    x^{-1}T \notin \fs
  }}
}\ 
\llbracket u \rrbracket^{-1} \bigl(\llbracket v \rrbracket \llbracket x^{-1}T \rrbracket\bigr)\Bigr) \text{ with } \mathit{diff}(x,y) = (u,v).
 \end{align*}
 
We now rewrite this representation of $\llbracket y^{-1} \mathit{maxsync}(S,T) \rrbracket$ based on the automata $\mathcal A, \mathcal B, \mathcal D_\gamma$. Let $\mathcal A\colon q_0^\mathcal A \xrightarrow{y} p$, and $\mathcal B\colon q_0^\mathcal B \xrightarrow{y} q$. Then $\llbracket L(\mathcal A_p) \rrbracket = \llbracket  y^{-1}S  \rrbracket$, and $\llbracket L(\mathcal B_q) \rrbracket = \llbracket  y^{-1}T  \rrbracket$.
Further, consider some $x$ compatible to $y$ such that $x^{-1}T \notin \fs$, and let $B\colon q_0^\mathcal B \xrightarrow{x} r$. Then $\llbracket L(\mathcal B_r) \rrbracket = \llbracket  x^{-1}T  \rrbracket$.
The lags of $x$ and of $y$ are at most $\gamma$, hence $\mathcal D_{\gamma}\colon (\varepsilon,\varepsilon) \xrightarrow{x \otimes y} (u,v)$ with $\mathit{diff}(x,y) = (u,v)$ according to \cref{lemma:kdelay}. Thus, we obtain
\begin{align*}
\llbracket y^{-1} \mathit{maxsync}(S,T)\rrbracket = \overbrace{
  \bigl(\llbracket L(\mathcal A_p) \rrbracket \cap \llbracket L(\mathcal B_q) \rrbracket\bigr) \setminus \Bigl( 
    \bigcup_{\mathclap{\substack{(r,u,v) \in G}}}\ 
    \llbracket u \rrbracket^{-1} \bigl(\llbracket v \rrbracket \llbracket L(\mathcal B_r) \rrbracket\bigr)
    \Bigr)}^{R_{p,q,G}},
\end{align*}
where $G = \{ (r,u,v) \mid \text{there is } x \text{ with } \mathcal B\colon q_0^\mathcal B \xrightarrow{x} r \in Q_\mathcal B \setminus Q_\mathcal B^{\mathsf{FS}} \text{ and }\mathcal D_{\gamma}\colon (\varepsilon,\varepsilon) \xrightarrow{x \otimes y} (u,v)\}$.
Note that $G$ is finite because $|u|,|v| \leq \gamma$ for all $(r,u,v) \in G$.
Furthermore, we define the transition profile $P(y)$ of $y$ as the triple $(p,q,G)$.
The profile $P(y)$ contains all necessary information to express $\llbracket y^{-1} \mathit{maxsync}(S,T)\rrbracket = R_{P(y)}$.
Let $\mathcal P = \{ P(y) \mid y \in \fse(T)\}$ be the set of all relevant profiles, and note that the set is finite.

We are ready to show that $\mathit{maxsync}(S,T)$ is regular.
Towards that we partition it into $\{ w \mid w^{-1}T \notin \fs\} \cap \mathit{maxsync}(S,T)$ and $\{ w \mid w^{-1}T \in \fs\} \cap \mathit{maxsync}(S,T)$ and show that both sets are regular.
The first set contains only synchronizations that do not have some $y \in \fse(T)$ as prefix, and the second set contains only synchronizations that have some $y \in \fse(T)$ as prefix.

The first set is regular, because if $w^{-1}T \notin \fs$, then $\mathit{lag}(w) \leq \gamma$, thus, we can rewrite first set as $\{ w\in S \cap T \mid w^{-1}T \notin \fs\}$.

The second set is described as
\[
\bigcup_{\mathclap{y \in \fse(T)}}\ y \cdot \bigl(y^{-1} \mathit{maxsync}(S,T)\bigr) = \bigcup_{\mathclap{y \in \fse(T)}}\ y \cdot \{ z \in y^{-1}T \mid \llbracket z \rrbracket \in \llbracket y^{-1}\mathit{maxsync}(S,T)\rrbracket\}.
\]
Using the definition of $\llbracket y^{-1}\mathit{maxsync}(S,T)\rrbracket$ from above we can express this as
\[
\bigcup_{\mathclap{(p,q,G) \in \mathcal P}}\ M_{p,q,G} \cdot 
\{ z \in L(\mathcal B_q) \mid \llbracket z \rrbracket \in R_{p,q,G}\},
\]
where $M_{p,q,G} = \{ y \in \fse(T) \mid P(y) = (p,q,G) \}$.
The set $M_{p,q,G}$ is regular because one can build an automaton that checks for in input $y$ whether $y \in \fse(T)$ and whether $P(y) = (p,q,G)$ by standard automaton constructions.
If $M_{p,q,G}$ is non-empty, then there is some $y \in \fse(T)$
with  $ R_{p,q,G} = \llbracket y^{-1}\mathit{maxsync}(S,T) \rrbracket$, which is recognizable by assumption.
Thus, the set $\{ z \in L(\mathcal B_q) \mid \llbracket z \rrbracket \in R_{p,q,G}\}$ is regular according to \cref{lemma:recognizable-subset}.

Hence, $\mathit{maxsync}(S,T)$ is a finite union of regular sets.

We have shown that $\mathit{maxsync}(S,T)$ is regular iff  $\llbracket y^{-1} \mathit{maxsync}(S,T)\rrbracket$ is recognizable for all $y \in \fse(T)$. We have seen that $y^{-1} \mathit{maxsync}(S,T)\rrbracket = R_{(p,q,G)}$ for $P(y) = (p,q,G)$ only depends on the profile of $y$, so there are only finitely many such relations. The relations used in the description of $R_{(p,q,G)}$ are all automatic, and hence $R_{(p,q,G)}$ is also automatic. Thus, it is decidable whether $R_{(p,q,G)}$ is recognizable.
\end{proof}

This completes the proof \cref{lemma:minmaxsync}.
\end{proof}

The decidability results presented next follow easily from the above lemma.

\begin{theorem}\label{prop:minmaxsync}
 It is decidable, given $S$ and $T \in \fsl$, whether $\mathit{minsync}(S,T)$, $\mathit{maxsync}(S,T)$, or $\mathit{allsync}(S,T)$ is a definition of $\llbracket S \rrbracket$ in $T$.
\end{theorem}

\begin{proof}
  The regularity of $\mathit{minsync}(S,T)$, $\mathit{maxsync}(S,T)$, and $\mathit{allsync}(S,T)$ is decidable by \cref{lemma:minmaxsync}.
  It is decidable whether $\llbracket S \rrbracket \subseteq \llbracket T \rrbracket$, because $S,T \in \fsl$ and $(\Sigma\Gamma)^*(\Sigma^* + \Gamma^*)$ is an effective canonical representation of $\fsl$.
  So, deciding whether $\llbracket S \rrbracket \subseteq \llbracket T \rrbracket$ can be reduced to deciding inclusion of regular languages.
\end{proof}

In combination with \cref{lemma:if-maxreg} we obtain the following theorem.
\begin{theorem} \label{thm:maxsyncTT-regular}
The resynchronized definability problem is decidable for given $S$ and $T \in \fsl$ with $\mathit{maxsync}(T,T)$ regular.
\end{theorem}

In \cref{sssec:prefixrec} we give an example for a class of languages $T$ such that $\mathit{maxsync}(T,T)$ is regular.

We finish with some relations between the regularity of the sets $\mathit{minsync}(S,T)$, $\mathit{maxsync}(S,T)$ and $\mathit{allsync}(S,T)$ for $S, T \in \fsl$.
In the next examples we show that the regularity of $\mathit{allsync}(S,T)$ does not imply the regularity of $\mathit{maxsync}(S,T)$ and vice versa.

In \cref{ex:max-not-reg}, we have given an example where $\mathit{allsync}(T,T)$ is regular and $\mathit{maxsync}(T,T)$ is not regular.
Recall that $\mathit{allsync}(T,T) = T$.

\begin{example}
 Consider $S = (\SigmaI\SigmaO)^*$ describing all pairs of words where the input and output component have the same length.
 Consider $T = (\SigmaI\SigmaO)^* + \SigmaI^*\SigmaO^*$.
 Clearly, $\mathit{maxsync}(S,T) = (\SigmaI\SigmaO)^*$ is regular.
 However, $\mathit{allsync}(S,T)$ is not regular, because it is of the form 
 \[
  (\SigmaI\SigmaO)^* \cup \{xy \in \SigmaI^*\SigmaO^* \mid x \in \Sigma^*, y \in \SigmaO^*, |x| = |y|\},
 \]
 and the latter set is not regular.
\end{example}

We show that the regularity of $\mathit{minsync}(S,T)$ implies the regularity of $\mathit{allsync}(S,T)$ and vice versa.

\begin{lemma}
 Given $S, T \in \fsl$, $\mathit{minsync}(S,T)$ is regular iff $\mathit{allsync}(S,T)$ is regular.
\end{lemma}

\begin{proof}
Assume $\mathit{allsync}(S,T)$ is regular.
Clearly, $\mathit{minsync}(S,T)$ can be expressed as 
\[
  \mathit{minsync}(\mathit{allsync}(S,T),\mathit{allsync}(S,T)).
\]
Thus, $\mathit{minsync}(S,T)$ is regular according to \cref{lemma:minreg}.
 
Assume $T' := \mathit{minsync}(S,T)$ is regular and note that $T' \subseteq T$. We have
\[
\mathit{allsync}(S,T) = T' \cup \underbrace{\{ w \in T \mid \text{there is } w' \in T' \text{ such that } w \succ_T w'\}}_U
\]
The set $U$ is regular by \cref{lemma:larger-sync-reg}.
\end{proof}

%% file: sections/def-unambiguous.tex
Based on the results in \cref{sssec:minmax} we can now show that the resychronized definability problem for automatic source relations and target languages of finite shiftlag is decidable if the target language is additionally unambiguous.

A regular language $T \subseteq (\SigmaIO)^*$ is called \emph{unambiguous} if $w_1 \neq w_2$ implies that $\llbracket w_1 \rrbracket \neq \llbracket w_2 \rrbracket$ for all $w_1,w_2 \in T$.

In general, unambiguity is undecidable for regular synchronization languages.
We obtain this by a simple reduction from a known result as follows.
As mentioned in \cref{sec:prelims}, for a regular synchronization language $S$ one can easily construct a transducer $\mathcal T$ such that $S(\mathcal T) = S$ and vice versa.
Hence, it is not difficult to see that $S$ is unambiguous iff $\mathcal T$ is unambiguous\footnote{Note that in the literature unambiguous is sometimes defined only with respect to the input.}, that is, for each $(u,v) \in \llbracket S \rrbracket = R(\mathcal T)$ there exists exactly one accepting run in $\mathcal T$ whose input is $u$ and output is $v$.
\cite{allauzen2011general} have shown that it is undecidable whether a transducer is unambiguous.
As a direct consequence we obtain that

\begin{lemma}
It is undecidable whether a given regular language $S \subseteq (\SigmaIO)^*$ is unambiguous.
\end{lemma}

However, unambiguity becomes decidable for regular synchronization languages with finite shiftlag.
\cite{allauzen2011general} have shown that it is decidable whether a transducer with finite lag is unambiguous which yields that it is decidable whether a regular synchronization language with finite lag is unambiguous.
We slightly generalize this. 

\begin{lemma}
It is decidable whether a given regular language $S \subseteq (\SigmaIO)^*$ with finite shiftlag is unambiguous.
\end{lemma}

\begin{proof}
Let $\mathcal A$ be a DFA for $S$. We consider an automatic relation over the input and output alphabets $\SigmaI \times Q_\mathcal A$ and $\SigmaO \times Q_\mathcal A$, respectively. This relation is defined by the regular synchronization language $S' \subseteq ((\SigmaIO) \times Q_\mathcal A)^*$ with finite shiftlag as 
\begin{align*}
 S':=\{ (a_1,q_0)(a_2,q_1)\dots(a_n,q_{n-1}) \mid &\ a_1\dots a_n \in L(\mathcal A) \text{ and }\\
 &\ q_0\cdots q_{n-1}\delta_\mathcal A(q_{n-1},a_n) \text{ is the run of $\mathcal A$ on } a_1\dots a_n \}.
\end{align*}

Let $S'' \subseteq ((\SigmaI \times Q_\mathcal A)(\SigmaO \times Q_\mathcal A))^*((\SigmaI \times Q_\mathcal A)^* + (\SigmaO \times Q_\mathcal A)^*)$ be the canonical representation of $S'$.
We show that the regular language $S$ is unambiguous iff $w_1 = (\sigma_1,r_0)\dots(\sigma_i,r_{i-1})$ and $w_2 = (\bar \sigma_1,s_0)\dots(\bar \sigma_j,s_{j-1})$ with $\sigma_1\cdots \sigma_i = \bar \sigma_1\cdots \bar \sigma_j$ implies that $r_0\cdots r_{i-1} = s_0\cdots s_{j-1}$ for all $w_1,w_2 \in S''$. 
This property can be checked given an NFA for $S''$.
If $S''$ does not satisfy the property one can guess letter-by-letter some word $u$ over $\SigmaIO$ and two different words $v,v'$ over $Q_\mathcal A$ and check that $u \otimes v \in S''$ and $u \otimes v' \in S''$ using an NFA for $S''$.

If $S$ is unambiguous then $w_1 = (\sigma_1,r_0)\dots(\sigma_i,r_{i-1})$ and $w_2 = (\bar \sigma_1,s_0)\dots(\bar \sigma_j,s_{j-1})$ with $\sigma_1\cdots \sigma_i = \bar \sigma_1\cdots \bar \sigma_j$ implies that $r_0\cdots r_{i-1} = s_0\cdots s_{j-1}$ for all $w_1$ and $w_2 \in S''$ since the automaton $\mathcal A$ is deterministic. 

For the other direction, we show that if $S$ is not unambiguous, then there exist $w_1,w_2 \in S''$ with $w_1 = (\sigma_1,r_0)\dots(\sigma_i,r_{i-1})$ and $w_2 = (\bar \sigma_1,s_0)\dots(\bar \sigma_j,s_{j-1})$ such that $\sigma_1\cdots \sigma_i = \bar \sigma_1\cdots \bar \sigma_j$ and $r_0\cdots r_{i-1} \neq s_0\cdots s_{j-1}$.
Assume that $S$ is not unambiguous, then there are $x,y \in S$ with $x \neq y$ and $\llbracket x \rrbracket = \llbracket y \rrbracket$, that is, two different synchronizations describe the same pair.
Let $w_1, w_2 \in S''$ with $w_1 = (\sigma_1,r_0)\dots(\sigma_n,r_{n-1})$ and $w_2 = (\bar \sigma_1,s_0)\dots(\bar \sigma_n,s_{n-1})$ denote their respective representations.
Since $x$ and $y$ encode the same pair, it is easy to see that $\sigma_1\cdots \sigma_n = \bar \sigma_1\cdots \bar \sigma_n$.
We show that $x \neq y$ implies $r_0\cdots r_{n-1} \neq s_0\cdots s_{n-1}$.

Towards a contradiction, we assume that $r_0\cdots r_{n-1} = s_0\cdots s_{n-1}$ which implies that $w_1 = w_2$.
Let 
\[(p_0,a_1,p_1)(p_1,a_2,p_2)\cdots (p_{n-1},a_n,p_n)\] be the run of $\mathcal A$ on $x$ and \[(q_0,b_1,q_1)(q_1,b_2,q_2)\cdots (q_{n-1},b_n,q_n)\] be the run of $\mathcal A$ on $y$.
Pick the first $i$ such that $a_i \neq b_i$, without loss of generalization, assume that $a_i \in \Sigma$ and $b_i \in \Gamma$.
Let $j$ be the smallest $j>i$ such that $j$ is a shift of $x$.
Let $u = a_i\cdots a_j$, clearly, $u\in\Sigma^*$.
Analogously, let $k$ be the smallest $k>i$ such that $k$ is a shift of $y$.
Let $v = b_i\cdots b_k$, clearly, $v\in\Gamma^*$.
Note that this implies that $a_{j+1} = b_i$ and $b_{k+1} = a_i$.
Furthermore, since $w_1 = w_2$, this also implies that $p_{i-1} = q_k$ and $q_{i-1} = p_j$.
Moreover, since $a_1\cdots a_{i-1} = b_1\cdots b_{i-1}$, we have that $\delta_\mathcal A(p_{0},a_1\cdots a_{i-1}) = p_{i-1} = q_{i-1} = \delta_\mathcal A(q_{0},b_1\cdots b_{i-1})$.
Thus, we can conclude that $p_{i-1} = q_{i-1} = q_k = p_j$.
Since $\delta_\mathcal A(p_{i-1},u) = p_j$ and $\delta_\mathcal A(q_{i-1},v) = q_k$, we obtain that $a_1\cdots a_{i-1}(u+v)^*a_{j+1}\cdots a_n \subseteq S$.
Clearly, this subset does not have finite shiftlag, thus, $S$ does not have finite shiftlag, which is a contradiction.
\end{proof}

Using the results from the previous section we obtain that

\begin{theorem}\label{prop:injectiveTarget}
The resynchronized definability problem is decidable for source languages with finite shiftlag and unambiguous target language with finite shiftlag.
\end{theorem}

\begin{proof}
Let $S, T \in \fsl$ and let $T$ be unambiguous.
Our goal is to decide whether $\llbracket S \rrbracket \in \textnormal{\textsc{Rel}}(T)$, i.e., whether there is a regular language $U \subseteq T$ such that $\llbracket U \rrbracket = \llbracket S \rrbracket$.
Assume, such a $U$ exists, then $U = \mathit{minsync}(S,T) = \mathit{maxsync}(S,T) = \mathit{allsync}(S,T)$, because $T$ is unambiguous.
Thus, $\llbracket S \rrbracket \in \textnormal{\textsc{Rel}}(T)$ iff $\mathit{minsync}(S,T)$ is regular and $\llbracket S \rrbracket = \llbracket \mathit{minsync}(S,T) \rrbracket$.
The first condition is decidable according to \cref{lemma:minmaxsync}, the second condition is decidable because $S, T \in \fsl$.
\end{proof}

%% file: sections/def-prefix.tex
In this section we describe a class of target synchronization languages $T$ that have the property that $\mathit{maxsync}(T,T)$ is regular. According to \cref{thm:maxsyncTT-regular}, the resynchronized definability problem is decidable for such targets (and source languages $S$ of finite shiftlag).

Based on that we then give an alternative proof for decidability of the problem whether a given (binary) automatic relation is prefix-recognizable (the decidability of this problem has already been shown by \cite{Choffrut:2014:DWS:2692050.2692051} for relations of arbitrary arity).

A set $U \subseteq (\SigmaI\SigmaO)^*$ is called \emph{$(\SigmaI\SigmaO)^*$-prefix closed} if $uab \in U$ implies $u \in U$ for all $u \in (\SigmaI\SigmaO)^*$ and all $a \in \SigmaI$, $b \in \SigmaO$.

The class of target languages that we consider contains the synchronization languages that are of the form $T = U\SigmaI^*\SigmaO^*$ for a regular $(\SigmaI\SigmaO)^*$-prefix closed set $U$.

\begin{lemma}\label{lem:prefix-closed}
Let $T = U\SigmaI^*\SigmaO^*$ for a regular $(\SigmaI\SigmaO)^*$-prefix closed set $U$. Then $\mathit{maxsync}(T,T)$ is regular.
\end{lemma}
\begin{proof}
  Let $w \in T$. Then $w = uxy$ with $u \in U$, $x \in \SigmaI^*$, and $y \in \SigmaO^*$.  Because $U$ is $(\SigmaI\SigmaO)^*$-prefix closed, we have that $w \in \mathit{maxsync}(T,T)$ iff
  \begin{itemize}
  \item $x = \varepsilon$, or $y = \varepsilon$, or
  \item $x = ax'$, $y = by'$, and $uab \notin U$.
  \end{itemize}
An automaton that accepts the words $w$ with this property can easily be built from an automaton for~$U$.
\end{proof}

A direct consequence of \cref{lem:prefix-closed} and \cref{thm:maxsyncTT-regular} is:
\begin{theorem} \label{thm:prefix-closed-dec}
The resynchronized definability problem is decidable for given $S,T \in \fsl$ where $T = U\SigmaI^*\SigmaO^*$ for a regular $(\SigmaI\SigmaO)^*$-prefix closed set $U$.
\end{theorem}

We now turn to prefix-recognizable relations. The standard definition of this class of relations uses a single alphabet $A$ as input and output alphabet. A relation over $A^* \times A^*$ is called \emph{prefix-recognizable} if it can be written in the form 
 \[
    \bigcup_{i=1}^n U_i (V_i \times W_i) := \{(uv,uw) \mid \text{$u \in U_i$, $v\in V_i$, and $w \in W_i$}\},
 \]
 for regular languages $U_i,V_i,W_i \subseteq A^*$ for each $i$.
 
A simple example of a prefix-recognizable relation is the lexicographical ordering on words over an ordered alphabet.

Prefix-recognizable relations where first studied by \cite{DBLP:journals/mst/AngluinH84} and \cite{DBLP:journals/jsyml/LauchliS87}.
The class of prefix-recognizable relations is a natural class that enjoys many nice properties.
An overview over these properties (for prefix-recognizable graphs) is given by \cite[Theorem 1]{Bl01}.
Furthermore, the graphs of prefix-recognizable relations have a decidable MSO theory which was shown by \cite{CAUCAL200379}.

\newcommand{\id}{\mathit{Id_A}}
\newcommand{\Tpr}{T_\mathit{pr}}
In order to capture prefix-recognizable relations in our setting, we we make the input and output alphabets disjoint by annotating input letters by $1$ and output letters by $2$, that is, $\SigmaI = \{1\} \times A$ and $\Gamma = \{2\} \times A$. And we say that a relation over $\SigmaI^* \times \SigmaO^*$ is prefix-recognizable if the relation over $A^* \times A^*$ that is obtained by removing the annotations $1$ and $2$ is prefix-recognizable.

Let $\id = \{(1,a)(2,a) \mid a \in A\}^* \subseteq (\SigmaI\SigmaO)^*$.
Then it is not hard to see that the prefix-recognizable relations are precisely those that can be defined by regular subsets of $\Tpr := \id\SigmaI^*\SigmaO^*$.
\begin{remark}
A relation over $\SigmaI^* \times \SigmaO^*$ is prefix recognizable iff it is in $\textnormal{\textsc{Rel}}(\Tpr)$.
\end{remark}

Since $\id$ is $(\SigmaI\SigmaO)^*$-prefix closed, we obtain the following corollary of \cref{thm:prefix-closed-dec}:

\begin{corollary}[\cite{Choffrut:2014:DWS:2692050.2692051}]
 It is decidable whether a given automatic relation is prefix-recognizable.
\end{corollary}

%% file: sections/conclusion.tex
We have considered uniformization and definability problems for subclasses of rational relations that are defined in terms of synchronization languages. Our results in combination with known results from the literature provide a systematic overview of the decidability and undecidability for the different variations of the problems. While this picture is almost complete, there are a few problems left open, as shown in \cref{tab:overview,tab:overviewdefinabilty}.

Furthermore, since we consider synchronization languages over $(\SigmaIO)^*$, our decidability results carry over to the more restricted case of synchronization languages over $\{1,2\}$. However, the undecidability results from the middle parts of the tables (where a synchronization language $T$ is given as input to the problem) do not carry over, in general. The undecidability results in \cref{sec:unif} also hold in the setting of synchronization languages over $\{1,2\}$ (in particular, \cref{thm:undec-aut-control} is formulated for such synchronization languages). For the undecidability results in \cref{tab:overviewdefinabilty} that are derived from \cref{prop:toRAT} it is open whether they carry over to synchronization languages over $\{1,2\}$.